\PassOptionsToPackage{unicode}{hyperref}
\PassOptionsToPackage{hyphens}{url}
\PassOptionsToPackage{dvipsnames,svgnames,x11names}{xcolor}
\documentclass[
]{article}

\usepackage{amsmath,amssymb}
\usepackage{iftex}
\ifPDFTeX
  \usepackage[T1]{fontenc}
  \usepackage[utf8]{inputenc}
  \usepackage{textcomp} 
\else 
  \usepackage{unicode-math}
  \defaultfontfeatures{Scale=MatchLowercase}
  \defaultfontfeatures[\rmfamily]{Ligatures=TeX,Scale=1}
\fi
\usepackage{lmodern}
\ifPDFTeX\else  
\fi
\IfFileExists{upquote.sty}{\usepackage{upquote}}{}
\IfFileExists{microtype.sty}{
  \usepackage[]{microtype}
  \UseMicrotypeSet[protrusion]{basicmath} 
}{}
\makeatletter
\@ifundefined{KOMAClassName}{
  \IfFileExists{parskip.sty}{%
    \usepackage{parskip}
  }{
    \setlength{\parindent}{0pt}
    \setlength{\parskip}{6pt plus 2pt minus 1pt}}
}{
  \KOMAoptions{parskip=half}}
\makeatother
\usepackage{xcolor}
\usepackage[top=20mm,left=30mm]{geometry}
\makeatletter
\ifx\paragraph\undefined\else
  \let\oldparagraph\paragraph
  \renewcommand{\paragraph}{
    \@ifstar
      \xxxParagraphStar
      \xxxParagraphNoStar
  }
  \newcommand{\xxxParagraphStar}[1]{\oldparagraph*{#1}\mbox{}}
  \newcommand{\xxxParagraphNoStar}[1]{\oldparagraph{#1}\mbox{}}
\fi
\ifx\subparagraph\undefined\else
  \let\oldsubparagraph\subparagraph
  \renewcommand{\subparagraph}{
    \@ifstar
      \xxxSubParagraphStar
      \xxxSubParagraphNoStar
  }
  \newcommand{\xxxSubParagraphStar}[1]{\oldsubparagraph*{#1}\mbox{}}
  \newcommand{\xxxSubParagraphNoStar}[1]{\oldsubparagraph{#1}\mbox{}}
\fi
\makeatother

\providecommand{\tightlist}{%
  \setlength{\itemsep}{0pt}\setlength{\parskip}{0pt}}\usepackage{longtable,booktabs,array}
\usepackage{calc} 
\usepackage{etoolbox}
\makeatletter
\patchcmd\longtable{\par}{\if@noskipsec\mbox{}\fi\par}{}{}
\makeatother
\IfFileExists{footnotehyper.sty}{\usepackage{footnotehyper}}{\usepackage{footnote}}
\makesavenoteenv{longtable}
\usepackage{graphicx}
\makeatletter
\newsavebox\pandoc@box
\newcommand*\pandocbounded[1]{
  \sbox\pandoc@box{#1}%
  \Gscale@div\@tempa{\textheight}{\dimexpr\ht\pandoc@box+\dp\pandoc@box\relax}%
  \Gscale@div\@tempb{\linewidth}{\wd\pandoc@box}%
  \ifdim\@tempb\p@<\@tempa\p@\let\@tempa\@tempb\fi
  \ifdim\@tempa\p@<\p@\scalebox{\@tempa}{\usebox\pandoc@box}%
  \else\usebox{\pandoc@box}%
  \fi%
}
\def\fps@figure{htbp}
\makeatother

\usepackage{titling}

\predate{\begin{center}\large Research Unit for Biostatistics\\ Department of Public Health, Aarhus University}

\usepackage{bbm}

\newcommand{\expect}[1]{ \mathbb{E} \left[ #1 \right] }

\newcommand{\cexpect}[2]{ \mathbb{E} \left[ #1 \: \big \vert #2 \: \right] }

\newcommand{\proba}[1]{ \mathbb{P} \left( #1 \right) }

\newcommand{\cproba}[2]{ \mathbb{P} \left( #1 \: \big \vert \: #2 \right) }

\newcommand{\setdef}[1]{ \left\{ #1 \right\} }

\newcommand{\indic}[1]{ \mathbbm{1}_{ \left\{ #1 \right\} } }

\newcommand{\indep}{\rotatebox[origin=c]{90}{$\models$}}

\DeclareMathOperator{\argmax}{argmax}

\newcommand{\history}[1]{\mathcal{H} \left( #1 \right)}

\newcommand{\dotbar}[1]{\bar{#1}_{\boldsymbol{\cdot}}}

\newcommand{\ExpressionPosteriorDensity}{
\left[
  \Phi
  \left(
    (-1)^{1-x}
    \frac{
      a + b \mu_{1+x} - \bar{y}_{\boldsymbol{\cdot}}^{(1)}
    }{
      \sqrt{\omega^2 + \sigma_{1+x}^2 b^2}
    }
  \right)
  \sigma_{1+x}
\right]^{-1}
\Phi \left( (-1)^{1-x} \frac{ (a + b \theta) - \bar{y}_{\boldsymbol{\cdot}}^{(1)} }{\omega} \right)
\phi \left(\frac{\theta - \mu_{1+x}}{\sigma_{1+x}}\right) }

\newcommand{\ExpressionPosteriorMean}{    \mu_{1+x}
    +
    (-1)^{1-x} b
    \frac{
	\sigma^2_{1 + x}
    }{\sqrt{\omega^2 + \sigma_{1+x}^2 b^2}}
    \frac{
    \phi
    \left(
    \frac{a + b \mu_{1 + x}  - \bar{y}_{\boldsymbol{\cdot}}^{(1)} }{\sqrt{\omega^2 + \sigma_{1+x}^2 b^2}}
    \right)
    }
    {
    \Phi
    \left(
    (-1)^{1-x}
    \frac{a + b \mu_{1 + x}  - \bar{y}_{\boldsymbol{\cdot}}^{(1)} }{\sqrt{\omega^2 + \sigma_{1+x}^2 b^2}}
    \right)
    }
}

\usepackage{listings}
\usepackage{lstbayes}

\definecolor{codegreen}{rgb}{0,0.6,0}
\definecolor{codegray}{rgb}{0.5,0.5,0.5}
\definecolor{codepurple}{rgb}{0.58,0,0.82}
\definecolor{backcolour}{rgb}{1,1,1}

\lstdefinestyle{mystyle}{
    backgroundcolor=\color{backcolour},   
    commentstyle=\color{codegreen},
    keywordstyle=\color{magenta},
    numberstyle=\tiny\color{codegray},
    stringstyle=\color{codepurple},
    basicstyle=\ttfamily\footnotesize,
    breakatwhitespace=false,         
    breaklines=true,                 
    captionpos=b,                    
    keepspaces=true,                 
    numbers=left,                    
    numbersep=5pt,                  
    showspaces=false,                
    showstringspaces=false,
    showtabs=false,                  
    tabsize=2
}

\usepackage{tikz}
\usepackage{pgfplots}

\usetikzlibrary{arrows, arrows.meta, positioning, decorations.pathreplacing, calligraphy}
\tikzset{>=latex'}

\makeatletter
\@ifpackageloaded{caption}{}{\usepackage{caption}}
\AtBeginDocument{%
\ifdefined\contentsname
  \renewcommand*\contentsname{Table of contents}
\else
  \newcommand\contentsname{Table of contents}
\fi
\ifdefined\listfigurename
  \renewcommand*\listfigurename{List of Figures}
\else
  \newcommand\listfigurename{List of Figures}
\fi
\ifdefined\listtablename
  \renewcommand*\listtablename{List of Tables}
\else
  \newcommand\listtablename{List of Tables}
\fi
\ifdefined\figurename
  \renewcommand*\figurename{Figure}
\else
  \newcommand\figurename{Figure}
\fi
\ifdefined\tablename
  \renewcommand*\tablename{Table}
\else
  \newcommand\tablename{Table}
\fi
}
\@ifpackageloaded{float}{}{\usepackage{float}}
\floatstyle{ruled}
\@ifundefined{c@chapter}{\newfloat{codelisting}{h}{lop}}{\newfloat{codelisting}{h}{lop}[chapter]}
\floatname{codelisting}{Listing}

\usepackage{amsthm}
\theoremstyle{plain}
\newtheorem{lemma}{Lemma}[section]
\theoremstyle{plain}
\newtheorem{proposition}{Proposition}[section]
\theoremstyle{remark}
\AtBeginDocument{}

\makeatother
\makeatletter
\@ifpackageloaded{tikz}{}{\usepackage{tikz}}
\makeatother
\makeatletter
\@ifpackageloaded{caption}{}{\usepackage{caption}}
\@ifpackageloaded{subcaption}{}{\usepackage{subcaption}}
\makeatother

\usepackage[round]{natbib}
\usepackage{bookmark}

\IfFileExists{xurl.sty}{\usepackage{xurl}}{} 
\hypersetup{
  pdftitle={Revisiting the apparent discrepancy between the frequentist and Bayesian interpretation of an adaptive design},
  pdfauthor={Simon Bang Kristensen; Erik Thorlund Parner},
  colorlinks=true,
  linkcolor={blue},
  filecolor={Maroon},
  citecolor={Blue},
  urlcolor={Blue},
  pdfcreator={LaTeX via pandoc}}

\title{Revisiting the apparent discrepancy between the frequentist and
Bayesian interpretation of an adaptive design}
\author{Simon Bang Kristensen \and Erik Thorlund Parner}
\date{}

\begin{document}
\maketitle

\begin{abstract}
\noindent
It is generally appreciated that a frequentist analysis of a group sequential trial must in order to avoid inflating type I error account for the fact that one or more interim analyses were performed. It is also to a lesser extent realised that it may be necessary to account for the ensuing estimation bias. The group sequential design is an instance of the more general concept of adaptive clinical trials where a study may change its design dynamically as a reaction to the observed data. There is a widespread perception that one may circumvent the statistical issues associated with the analysis of an adaptive clinical trial by performing the analysis under a Bayesian paradigm. The root of the argument is that the Bayesian posterior is perceived as being unaltered by the data-driven adaptations. We examine the claim that the posterior distribution is unaltered by adaptations by analysing a simple trial with a single interim analysis. We approach the interpretation of the trial data under both a frequentist and Bayesian paradigm with a focus on estimation. The conventional result is that the interim analysis impacts the estimation procedure under the frequentist paradigm, but not under the Bayesian paradigm, which may be seen as expressing a "paradox" between the two paradigms. We argue that this result however relies heavily on what one would define as the universe of relevant trials defined by first samples of the parameters from a prior distribution and then the data from a sampling model given the parameters. In particular, in this set of trials, whether a connection exists between the parameter of interest and design parameters. We show how an alternative interpretation of the trial yields a Bayesian posterior mean that corrects for the interim analysis with a term that closely resembles the frequentist conditional bias. We conclude that the role of auxiliary trial parameters needs to be carefully considered when constructing a prior in an adaptive design.
\end{abstract}

\section{Introduction}\label{sec-introduction}

An adaptive design in the context of clinical trials is a study design
that may change, or adapt, dynamically as the study data is accrued. One
such commonly applied adaptation allows for interim analyses during the
study period where the hitherto collected data is analysed and the
results used to stop the trial if there is clear evidence in favor of a
specific treatment (stopping for efficacy) or if it is deemed unlikely
that the trial will terminate with a conclusive outcome even if carried
to its fruition (stopping for futility). Such group sequential designs
(\citet{jennisonGroupSequentialMethods1999}) may be viewed as a specific
subset of adaptive designs. In contrast, a \emph{fixed design} is a
design that is run in accordance with a prespecified study protocol
until the data collection is completed. Adaptations of the design come
with obvious advantages
(e.g.~\citet{pallmannAdaptiveDesignsClinical2018a}). Logistic advantages
include for example the possibility to stop a trial early thereby
requiring fewer patients and thus lowering on average the overall costs
associated with the study. Arguable, adaptive designs may also be said
to be more ethical, as the presence of interim analyses may allow a
clearly effective treatment to benefit patients with less delay than if
the trial had been required to run until its fixed termination while
avoiding the subjection of trial participants to an inferior treatment.

An adaptive design may also entail certain practical difficulties
(\citet{pallmannAdaptiveDesignsClinical2018a}). The main challenges are
however those that concern statistical inference. One issue that has
traditionally received much attention is the inflation of type I error.
For example, a trial employing multiple interim analyses at a given
level of significance will have a higher overall type I error than this
significance level, which may be seen simply as a problem of multiple
testing. Several procedures exist to counter this inflation of type I
error. A less well understood (\citet{bretz2009adaptive},
\citet{bauerTwentyfiveYearsConfirmatory2016}) problem is the consequence
of the adaptive design to estimation and particular to estimation bias
-- a study allowing to stop if there is evidence of a large treatment
effect will invariably overestimate the effect of the treatment on
average.

A line of inquiry that has received some attention is whether the
statistical issues associated with an adaptive design may be ameliorated
by approaching the analysis of the trial under a Bayesian framework. A
recent Lancet review of Bayesian methods for clinical trials
(\citet{goligherBayesianStatisticsClinical2024}) emphasises these
possibilities and argue that ``\emph{issues such as the original planned
sample size and stopping rule {[}\ldots{]} do not affect the Bayesian
posterior distribution}''. A blog by Frank Harrell (\citet{harrellblog})
takes a similar view that ``\emph{the stopping rule is unimportant when
interpreting the final evidence. Earlier data looks are irrelevant.}''.
Harrell frames this in terms of the calibration property of Bayesian
inference: Suppose that we sample parameters from our prior distribution
and data from a sampling model conditional on the parameters. This gives
us a universe of relevant trials, a universe that is defined by our
chosen prior distribution. We then recognise our observed data among
some of these relevant trials and ask: Among those relevant trials that
obtained the same data, what was the distribution of the parameters.
This frequency distribution is exactly the Bayesian posterior
distribution or we could say that the posterior is calibrated to the
distribution. This apparent property of Bayesian inference is in stark
contrast to the challenges concerning the frequentist concepts of type I
error and estimation bias as noted above. Seemingly, the implication is
that we can harvest the evident advantages of the adaptive design and
circumvent the statistical challenges providing that we analyse the
trial under the Bayesian framework.

A related but distinct issue arises in the context of the discussion
whether Bayesian inference is immune to the effect of selection,
i.e.~that the posterior distribution for a parameter is not altered even
if the specific parameter was chosen as a target of inference following
a selection process. The problem was studied in
\citet{dawidSelectionParadoxesBayesian1994} and again in
\citet{sennNoteConcerningSelection2008} and
\citet{mandelSelectionBiasConflict2009}. These investigations agree that
this apparent property of Bayesian inference is an artifact of a prior
on the parameter space that specify the parameters as being independent.
\citet{harvilleBayesianInferenceUnaffected2022} studied the problem more
generally and argued for basing the posterior distribution on that
conditional to selection, and further shows how this may also be
recharacterised in terms of using a different prior distribution.

Observed discrepancies between the frequentist and Bayesian framework
have conventionally been termed paradoxes in the literature, so for
example by \citet{dawidSelectionParadoxesBayesian1994}. As also noted by
several authors (e.g.~\citet{sennNoteConcerningSelection2008}), these
discrepancies are not true paradoxes as they are clearly anticipated by
the mathematics, but may be viewed as more apparent paradoxes that
describe a discrepancy where one might commonly be perceived to not
exist. Nevertheless, we will follow this tradition and speak of
``paradoxes'' for the remainder of the paper. Two classic paradoxes that
contrast frequentist and Bayesian inference are given by John Pratt in
the discussion paper \citet{savageFoundationsStatisticalInference1962}
in terms of measurement instruments, and by
\citet{lindleyStatisticalParadox1957} in the form of a disagreement
between the frequentist hypothesis test and Bayesian posterior
distribution.

In the following, we analyse the inferential repercussions of an
adaptive design under the frequentist and Bayesian paradigm by focusing
on a simple study described in \citet{sennStatisticalIssuesDrug2008}
(Chapter 19) that may be viewed as a paradox in the above sense.

Suppose that two investigators, \(A\) and \(B\), will run a simple study
to perform inference about an unknown parameter \(\Theta\). The study
will first collect \(n\) observations, and is then given the option to
stop if the average of the first \(n\) observations exceeds some
prespecified \(\Psi\). Otherwise, another \(n\) observations will be
collected. The two investigators will run the study in each their own
way. Investigator \(A\) will collect all \(2n\) data points regardless
of the outcome of the interim analysis of the data with \(n\)
observations. Investigator \(B\), however, will allow the trial to stop
after the first \(n\) observations and otherwise collect \(2n\)
observations.

We now imagine that the two investigators run their study and obtain
\emph{exactly} the same data and further that this data did not lead
investigator \(B\) to stop their study early. Thus, the two researchers
have arrived at exactly the same data sets by somewhat different means.
This begs (at least) the following questions,

\begin{enumerate}
\def\labelenumi{\arabic{enumi}.}
\tightlist
\item
  What may \(A\) and \(B\), respectively, infer about \(\Theta\) based
  on the obtained data?

  \begin{enumerate}
  \def\labelenumii{\arabic{enumii}.}
  \tightlist
  \item
    Under a Bayesian paradigm?
  \item
    Under a frequentist paradigm?
  \end{enumerate}
\item
  Why would \(A\) and \(B\) elect to conduct their studies differently?
\end{enumerate}

We recognise the design of investigator \(B\) as a simple group
sequential design with one interim analysis allowing the study to
terminate for efficacy halfway to the maximal sample size of \(2n\). The
trial run by investigator \(A\) may be viewed as a corresponding fixed
design.

As we shall elaborate on in the following, the standard conception is
that the two frequentist investigators must approach their inference
differently with investigator B having to compensate for their
intentions of stopping the trial, while the two investigators obtain the
same inference under a Bayesian paradigm. This apparent discrepancy is
why the situation may be characterised as a ``paradox''.

~\newline \noindent The remainder of the paper is structured as follows.
We first establish the setup and notation for the paper and then derive
a bound on the adaptations that may be introduced into a design without
altering the likelihood function by introducing what we call a
well-behaved design. We further motivate the posterior distribution in
terms of its calibration to a universe of relevant trials. We then
analyse the paradox above in the form of a simple adaptive design. We
show that in a setup which includes a nuisance parameter, the
discrepancy between the frequentist and Bayesian interpretation rests on
the specification of the parameter prior and in particular on the
dependence between the treatment effect of interest and the design
nuisance parameter. Throughout, we illustrate the different
interpretations of a small simulated example of the paradox. We finally
discuss the implications for the interpretation of an adaptive trial.

\section{Notation and general theory}\label{notation-and-general-theory}

\subsection{Frequentist and Bayesian
inference}\label{sec-setup-and-notation}

In the following, we discuss inference under a frequentist and Bayesian
paradigm and will need to define their interrelationship. This is
approached in the spirit of e.g.~\citet{EfronBradley2016CASI}, where
Figure~\ref{fig-twodim-inference} below illustrates the difference
between the two paradigms. Frequentist inference is concerned with the
conditional distribution of the data \(\mathcal{D}\) given a specific
value of the parameters \(\mathcal{P}\), i.e.~the distribution of data
under repeated sampling from the same underlying distribution as
parametrised by \(\mathcal{P}\). Bayesian inference may be viewed as
orthogonal in the sense, that it is concerned with the conditional
distribution of the parameters given the data. The Bayesian paradigm
places a prior distribution on the parameters thus effectively assuming
a full distribution on the two-dimensional space in the figure.
Inference concerning the posterior distribution of the parameter given
the data is achieved from the sampling model
\(\mathcal{D} \mid \mathcal{P}\) along with the prior distribution on
\(\mathcal{P}\) through Bayes formula, which states that the posterior
density is proportional to the product of the sampling density and the
prior density, i.e.
\begin{equation}\phantomsection\label{eq-bayes-formula}{
    f_{\mathcal{P} \mid \mathcal{D}} (\mathcal{p} \mid \mathcal{d}) = \frac{f_{\mathcal{D} \mid \mathcal{P}} (\mathcal{d} \mid \mathcal{p}) f_{\mathcal{P}}(\mathcal{p})}{\int f_{\mathcal{D} \mid \mathcal{P}} (\mathcal{d} \mid \mathcal{p}) f_{\mathcal{P}}(\mathcal{p}) \: d \mathcal{p}} 
}\end{equation} The integral in the denominator is a normalising
constant to ensure that the posterior density integrates to one.

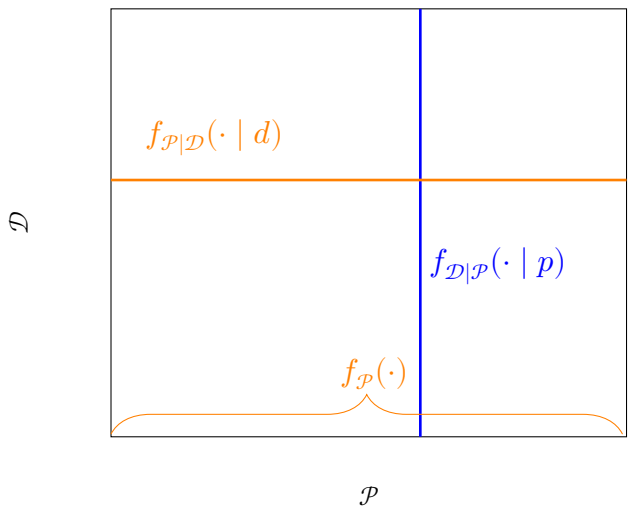
\begin{figure}

\centering{

\begin{tikzpicture}
  \begin{axis}[
    xmin=0, xmax=1,
    ymin=0, ymax=1,
    xtick=\empty, ytick=\empty,
    xlabel={$\mathcal{P}$},
    ylabel={$\mathcal{D}$}
    ]
    \addplot [ycomb, blue, line width=1pt] coordinates {(0.6, 1)};
    \addplot [xcomb, orange, line width=1pt] coordinates {(1, 0.6)};
    \node at (axis cs:0.2,0.7) {\large\color{orange}$f_{\mathcal{P} \mid 
        \mathcal{D}} (\cdot \mid d)$};
    \node at (axis cs:0.75,0.4) {\large\color{blue}$f_{\mathcal{D} \mid 
        \mathcal{P}} (\cdot \mid \mathcal{p})$};
  \end{axis}
  \node at (3.5,0.85) (prior) {\large\color{orange}$f_{\mathcal{P}} (\cdot)$};
  \draw [orange, decorate, decoration={brace, amplitude=15pt, raise=1pt}] (0,0) -- (6.8,0);
\end{tikzpicture}

}

\caption{\label{fig-twodim-inference}Graphical representation of the
relationship between frequentist and Bayesian inference. The horisontal
axis represents a parameter while the vertical axis represents the data.
The vertical, blue line is the direction of the frequentist inference,
the conditional variation of the data given the parameter. Analogously,
the horisontal, orange line is the conditional distribution of the
parameter given the data (i.e.~the posterior distribution) with which
Bayesian inference is concerned. Along the horisontal axis, the presence
of a prior distribution on the parameter space in the Bayesian inference
is represented by the orange \(f_{\mathcal{P}}\).}

\end{figure}%

We write \(\mathcal{P}\) for the unknown parameters and write
\(\mathcal{p}\) for an instance of \(\mathcal{P}\) and will generally
use lower-case letters to denote a realisation of a random variable
written in the corresponding upper-case.

The posterior probability in the Bayesian framework may be given a
frequentist interpretation as follows. Draw a parameter \(\mathcal{P}\)
from the prior distribution with density \(f_{\mathcal{P}}\) and draw a
set of data from the conditional distribution
\(f_{\mathcal{D} \mid \mathcal{P}}\). Repeating this many times yields a
frequency distribution of observations and parameter values and if we
focus on those parameters that are given by a specific observation of
data we obtain the density \(f_{\mathcal{P} \mid \mathcal{D}}\). That
this is exactly the posterior distribution follows from Bayes formula,
which states that the distribution resulting from the two-stage sample
procedure described above is proportional to the posterior distribution.
This property is sometimes said to state that the Bayesian posterior is
\emph{calibrated} and we will refer to this description of the posterior
as its \emph{calibration representation}. This is a frequency
distribution but, importantly, in a world defined by the specific choice
of prior, and thus it will not coincide with the conventional frequency
distribution that arises from repeated sampling for fixed
\(\mathcal{P}\) as defined by an experimental setting in the real world.

We will need the concept of the ``true'' parameters \(\mathcal{p}_0\),
which are the specific parameters that govern the data generating
process. In a Bayesian setting, this may be motivated through a
so-called Bernstein--von-Mises theorem (\citet{gelman2013bayesian},
Chapter 4) which states that,

\[
\sqrt{n} (\mathcal{P} - \mathcal{p}_0) \mid \mathcal{D} \overset{\sim}{\longrightarrow} N \left( 0, I^{-1}(\mathcal{p}_0) \right) ,
\]

where the convergence is in distribution and \(I(\mathcal{p})\) is the
Fisher information at \(\mathcal{p}\). Here \(\mathcal{p}_0\) are the
data-generating parameters. That is, if \(g(\cdot ; \chi)\) are
densities belonging to a family of distributions parametrised by
\(\chi\), then \(f_{\mathcal{D}} = g(\cdot ; \mathcal{p}_0)\). On a more
technical note the statement of the theorem is that \(\mathcal{p}_0\)
are the parameters governing the data generating process if the Bayesian
model is sufficiently rich to include these in the posterior
distribution. Otherwise, \(\mathcal{p}_0\) will be the parameters that
minimise the distance to the data-generating parameters in a
Kullback-Leibler sense. The theorem may also be interpreted as stating
that for large \(n\), the posterior distribution approximately resembles
a normal distribution with mean \(\mathcal{p}_0\) and variance
\(\frac{1}{n} I^{-1} (\mathcal{p}_0)\). This convergence, apart from the
technical sense just noted, is independent of the choice of prior
implying that the impact of the choice of prior decreases as more data
is obtained.

Having established these concepts, we may give meaning to various
frequentist concepts under the Bayesian paradigm. For instance, if
\(\hat{\mathcal{p}} = \hat{\mathcal{p}} (\mathcal{D})\) is an estimator
of \(\mathcal{p}_0\), viewed as a function of the data, the bias of the
estimator is
\(\cexpect{\hat{\mathcal{p}} - \mathcal{P}}{\mathcal{P} = \mathcal{p}_0} = \cexpect{\hat{\mathcal{p}}(\mathcal{D})}{\mathcal{P} = \mathcal{p}_0} - \mathcal{p}_0\).
Similarly if \(\phi\) is a statistical test, i.e.~an indicator function
for rejection, for the hypothesis \(\mathcal{P} = c\) for some constant
\(c\), then \(\cexpect{\phi}{\mathcal{P} = c}\) is the type 1 error rate
of the test.

Returning to the paradox described in Section~\ref{sec-introduction}, a
conventional approach would take \(\mathcal{P} = \Theta\) to be the
parameter of interest. We also explore the option of
\(\mathcal{P} = (\Theta, \Psi)\) being the combined parameter of
interest including the design parameter represented by the stopping
threshold \(\Psi\).

\subsubsection{Adaptive designs}\label{sec-adaptive-designs}

We consider a study consisting of \(N\) observations of a design
variable \(X\) and an outcome \(Y\). Thus, the study records the data
\((\mathbf{X}, \mathbf{Y})\), where we use the notation
\(\mathbf{X} = (X_1, \ldots, X_N)\) for a vector (and similarly for
\(\mathbf{Y}\)). We may as an example represent a two-arm study by
letting \(X_i \in \setdef{1 , 2}\) be an indicator of the arm to which
individual \(i\) was assigned and \(Y_i\) is the recorded outcome of the
same individual. As another example, suppose a study is conducted using
a poor data recording process which will delete an observation with some
small probability. Let \(X_i \in \setdef{0, 1}\) be an indicator for
recording of the information and set \(Y_i = \cdot\) if \(X_i = 0\) for
some arbitrary value \(\cdot\), and \(Y_i = \tilde{Y}_i\) when
\(X_i = 1\) and \(\tilde{Y}_i\) is the original outcome (our notation
here mirrors the treatment of missing data in \citet{rubin1976}). This
example may seem a bit contrived, but we will use this notation to
represent studies with interim analyses where we will set
\(Y_i = \cdot\) for any observation \(i\) that has not been observed due
to the trial having stopped. The random sample size of the trial is
\(\sum_{i=1}^N X_i\).

The study is conducted to draw inference about the unknown parameters
\(\Theta\), which are perceived as the parameters governing the
conditional distribution \(Y_i \mid X_i\). We write \(\theta\) for an
instance of \(\Theta\) and let \(\theta_0\) be the specific parameters
that govern the data generating process (the `true' parameter values).

When \(W\) is a random variable we write \(f_W\) for its density
function. If \(\mathbf{W} = (W_1, \ldots, W_K)\) is a vector, we write
\(\history{W_i} = (W_1, \ldots, W_{i-1})\) for the \emph{history} of
\(W\) up to \(i\) for \(i=1, \ldots, K\), with the convention that
\(\history{W_1}\) is the empty vector. Note the assumption that the
subscript \(i\) indexes the observations in the order in which they were
collected.

The likelihood function is the sampling distribution of the observed
data perceived as a function of the parameters,

\[
  L(\theta) = f_{(\mathbf{X}, \mathbf{Y}) \mid \Theta} (\mathbf{x}, \mathbf{y} \mid \theta) .
\]

By a fixed design we mean a design in which the design variable \(X_i\)
(for every \(i = 1, \ldots, N\)) does not depend on any other variable
in the design, so that the \(X\)'s could in principle be fixed before
the beginning of a study. As an example, we may consider the two arm
study where \(X_i \in \setdef{1, 2}\) indicates the treatment allocated
to participant \(i\). If \(X_i\) is determined be simple randomisation
(i.e.~flipping a coin), this would be an example of a fixed design, and
the allocations in the study would usually be generated prior to
beginning the study in a (blinded) randomisation list. Collecting
outcomes independent of each other, the likelihood function under the
fixed design is,

\begin{equation}\phantomsection\label{eq-fixed-likelihood}{
L(\theta) \overset{\text{fixed}}{=} \prod_{i = 1}^N f_{Y_i \mid X_i, \Theta} \left( y_i \mid x_i, \theta \right) .
}\end{equation}

We use \emph{directed acyclic graphs} (DAGs) as a convenient way to
represent (conditional) independence assumptions between random
variables (e.q. \citet{greenlandCausalDiagramsEpidemiologic1999}). A DAG
will include both parameters and data, so that the DAG describes the
joint distribution imposed by Bayesian inference as described in
connection with Figure~\ref{fig-twodim-inference}. For clarity, we have
inscribed parameters in a rectangle. The fixed design is illustrated in
the DAG in Figure~\ref{fig-fixed-design}.

\begin{figure}

\centering{

\pandocbounded{\includegraphics[keepaspectratio]{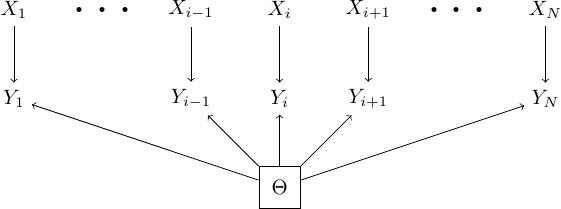}}

}

\caption{\label{fig-fixed-design}Fixed design DAG.}

\end{figure}%

We will say that a design is \emph{well-behaved} (W-B) if,

\begin{itemize}
\tightlist
\item
  \(Y_i\) is conditionally independent of the histories
  \(\history{X_i}\) and \(\history{Y_i}\) given the current design
  \(X_i\) and the parameter of interest \(\Theta\), and,
\item
  \(X_i\) is conditionally independent of the parameter of interest
  \(\Theta\) given the histories \(\history{X_i}\) and
  \(\history{Y_i}\).
\end{itemize}

The following result will be key to the discussion below (a proof is
given in Appendix \ref{sec-proof-prp-wb-likel}). Similar distinctions
are made in \citet{dawidIdentifyingConsequencesDynamic2010} and
\citet{kristensenAnalysingBiasIntroduced2025}.

\begin{proposition}[]\protect\hypertarget{prp-wb-likel}{}\label{prp-wb-likel}

If the design is well-behaved, the likelihood function is proportional
to the likelihood under the fixed design.

\end{proposition}

\subsubsection{Some examples of designs}\label{sec-design-examples}

It is easy to verify that the fixed design is well-behaved. Another
example of a well-behaved design is given by the DAG in
Figure~\ref{fig-wb-design}. Here, we will let \(X_i\) be a function of
\(X_{i-1}\) and \(Y_{i-1}\) or, more generally, of the histories of
\(X_i\) and \(Y_i\). This will include for example so-called
dose-adaptive designs, where the dose \(X\) of the next individual
included in the study is determined from the previously allocated doses
and the previous responses.

\begin{figure}

\centering{

\pandocbounded{\includegraphics[keepaspectratio]{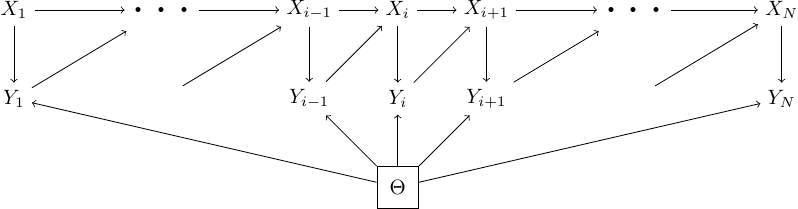}}

}

\caption{\label{fig-wb-design}Example of a well-behaved design, DAG.}

\end{figure}%

An example of a design that is not well-behaved could be the following.
Observe \(N\) binary responses \(\tilde{Y}_i \in \setdef{0, 1}\) and
define

\[
n_0 = \argmax_{i = 1, \ldots, N} \frac{1}{i} \sum_{j=1}^i \tilde{Y}_j
\]

to be the index that maximises the average response up to this index.
Set \(X_i = \indic{i \leq n_0}\) and define the outcomes \(Y_i\) as
\(\tilde{Y}_i\) if \(X_i = 1\), and \(Y_i = \cdot\) if \(X_i = 0\). In
other words, we run a study of size \(N\) but discard data after
\(n_0\), effectively pretending as though we had discontinued the study
once we reached the index that we know will lead to the largest average
response.

For this design it is apparent that for example the second condition for
a W-B design is not fulfilled, e.g.~\(X_{N-1}\) will conditional on the
history of \(X\)'s and \(Y\)'s up to \(N-1\) still depend on \(\Theta\)
though \(Y_N\) in violation of the second requirement for a W-B design.
While the design is not well-behaved the posterior distribution is still
calibrated. However, considering two investigators, one who performs a
fixed design and one who performs the design discarding all observations
after \(n_0\), the two posteriors will generally not be calibrated to
the same distribution, a point to which we will return in the discussion
below.

\section{Revisiting the paradox}\label{revisiting-the-paradox}

We now return to the paradox sketched in Section~\ref{sec-introduction}
above to address the two questions posed there. In the notation
introduced in the previous section, the study in the paradox may be
described as follows. Collect \(n\) outcomes \(Y_1, \ldots, Y_n\). The
study is given the option to stop which it will do if
\(\frac{1}{n} \sum_{i = 1}^n Y_i > \Psi\) for some prespecified
\(\Psi\), and otherwise another \(n\) observations
\(\tilde{Y}_{n+1}, \ldots, \tilde{Y}_{2n}\) will be collected. Here we
take \(\Psi\) to be a parameter (we return to this choice below). We
have design variables \(X_1, \ldots, X_{2n}\) where \(X_i = 1\) for
every \(i=1, \ldots, n\) and
\(X_i = \indic{\frac{1}{n} \sum_{i = 1}^n Y_i \leq \Psi}\) for
\(i=n+1, \ldots, 2n\). The last \(n\) observations are
\(Y_i=\tilde{Y}_i\) when \(X_i = 1\) and \(Y_i=\cdot\) when \(X_i=0\)
for every \(i=n+1,\ldots, 2n\). Since the design variables are constant
and equal to one for the first \(n\) observations and constant for the
remaining \(n\) observations, we may collapse the notation and simply
denote by \(X\) the indicator for continuation. We note that
investigator A is effectively running a fixed design study, and so we
may write
\(X = \indic{\frac{1}{n} \sum_{i=1}^n Y_i \leq \Psi} \indic{P = B} + \indic{P=A}\).
The random sample size is \((1+X)n\).

\begin{figure}

\centering{

\pandocbounded{\includegraphics[keepaspectratio]{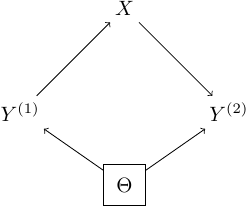}}

}

\caption{\label{fig-simple-case}Simple DAG for the paradox.}

\end{figure}%

The answer to Question 1 is supplied by the literature, and we shall
include its answer below for the sake of completeness.

\subsection{Question 1 and the conventional
analysis}\label{sec-q1-and-conventional-analysis}

Under a frequentist paradigm, the design of investigator B is readily
recognised as a group sequential trial including an interim analysis to
allow for stopping for efficacy and the accompanying literature
describes the consequences to both type 1 error rate and estimators
(\citet{jennisonGroupSequentialMethods1999}). To make this slightly more
concrete we consider estimation in the scenario where the sampling model
is \(Y_i \mid X, \Theta \sim N(\Theta, \sigma^2)\) for
\(i = 1, \ldots, (1+X)n\), where \(\sigma > 0\) is known. As already
noted the design is well-behaved and thus both investigators would
perform maximum likelihood using the same likelihood function, the MLE
being
\(\hat{\theta} (X, Y) = \frac{1}{(1+X)n} \sum_{i = 1}^{(1+X)n} Y_i\).
One may derive the expectation of the estimator under repeated sampling
(see Appendix \ref{sec-deriv-bias-expression}) as,

\begin{equation}\phantomsection\label{eq-bias-freq}{ 
\cexpect{\hat{\theta} (X, Y)}{\Theta, \Psi} = \Theta + \frac{\sigma}{2 \sqrt{n}} \phi \left(\frac{\sqrt{n}}{\sigma} \left[ \Psi - \Theta \right] \right)
}\end{equation}

where \(\phi\) is the density function of a standard normal
distribution. Letting \(\Phi\) denote the cumulative distribution
function of the standard normal distribution, we also note that the
conditional biases are given by
\begin{equation}\phantomsection\label{eq-conditional-bias-x0}{
\begin{aligned}
  \cexpect{\hat{\theta} (X, Y)}{\Theta, \Psi, X = 0} - \Theta
  &=
     \frac{\sigma}{\sqrt{n}} \frac{\phi \left(\frac{\sqrt{n}}{\sigma} \left[\Psi - \Theta \right]\right)}{\Phi \left(-\frac{\sqrt{n}}{\sigma} \left[\Psi - \Theta \right]\right)} ,
\end{aligned}
}\end{equation} and,
\begin{equation}\phantomsection\label{eq-conditional-bias-x1}{
\begin{aligned}
  \cexpect{\hat{\theta} (X, Y)}{\Theta, \Psi, X = 1} - \Theta
  &=
    - \frac{\sigma}{2\sqrt{n}} \frac{\phi \left(\frac{\sqrt{n}}{\sigma} \left[\Psi - \Theta \right]\right)}{\Phi \left(\frac{\sqrt{n}}{\sigma} \left[\Psi - \Theta \right]\right)} .
\end{aligned}
}\end{equation} Thus, the bias conditional on stopping \(X=0\) is of an
opposite sign from the conditional bias given continuation \(X=1\) and
the former will usually be numerically larger than the latter (unless
the continuation probability is less than \(1/3\)). The marginal bias in
(\ref{eq-bias-freq}) is a weighted average of the two conditional biases
and is smaller than either.

If both investigators were Bayesians, the paradox is usually formulated
by assuming that they would apply the same prior distribution on
\(\Theta\). If the investigators, as in the frequentist scenario above,
further agree on the sampling model, it follows from Bayes theorem and
Proposition~\ref{prp-wb-likel} that they would obtain exactly the same
posterior distribution and thus identical, Bayesian inference.

For the sake of completeness, we supply a standard Bayesian analysis in
the spirit of \citet{spiegelhalterBayesianApproachesRandomized1994},
which would proceed as follows. Taking the sampling model to be as
above, we would specify a prior distribution on the treatment effect
\(\Theta\). For analytical simplicity, we might take the prior to be
conjugate to the sampling model and specify
\(\Theta \sim N(\mu, \tau^2)\) and in line with the formulation of the
paradox we suppose that this prior is used by both investigators. By
Bayes formula, the posterior density is proportional to the product of
the likelihood and the prior density, and as noted above we arrive at
the same posterior for the two investigators. Standard calculations
(also see derivations in Appendix \ref{sec-derivations}) show that the
posterior is, \begin{equation}\phantomsection\label{eq-posterior-q1}{ 
    \Theta \mid X, Y \sim
    N
    \left( 
        \mu_{1+X}
        ,
        \sigma^2_{1+X}
    \right)
}\end{equation} with \[
\mu_{1+x}
=
\bar{y}_{\boldsymbol{\cdot}} \frac{\tau^2}{\tau^2 + \frac{\sigma^2}{(1+x)n}} + \mu \frac{\frac{\sigma^2}{(1+x)n}}{\tau^2 + \frac{\sigma^2}{(1+x)n}} 
\] and, \begin{equation}\phantomsection\label{eq-q2-mu-sigma-x}{
\sigma^2_{1+x}
=
\frac{\frac{\sigma^2}{(1+x)n} \tau^2}{\tau^2 + \frac{\sigma^2}{(1+x)n}}
}\end{equation} We note that \(\frac{\sigma^2}{(1 + x) n}\) is the
squared standard error (in a frequentist sense) of
\(\dotbar{y} = \sum_{i=1}^{(1+x)n} Y_i / ((1+x)n)\) so that the
posterior mean may be perceived as a weighted average between the
data-supplied estimate \(\dotbar{y}\) and the prior knowledge about the
mean as represented by \(\mu\), the weight being the relative difference
between the standard error and the prior variance \(\tau^2\). If
\(\tau \uparrow \infty\) an improper prior ensues that does not depend
on \(\mu\) and the posterior mean agrees with the frequentist estimate
\(\dotbar{y}\). This is in line with the more general observation, that
if a flat (i.e.~non-informative) prior is used, then maximising the
posterior amounts to maximising the likelihood, so that the posterior
mode estimate coincides with the maximum likelihood estimate, in this
case \(\dotbar{y}\). When a proper prior is used with \(\tau > 0\), the
estimate \(\dotbar{y}\) is shrunk towards the prior mean \(\mu\). If a
large study is performed, the standard error is small so that more
weight is placed on the estimate \(\dotbar{y}\) and the shrinkage is
small, while larger shrinkage occurs in a smaller study. Taking the
posterior mean as an estimator, it is biased in the frequentist sense:
Under repeated sampling of the data given the parameter, the mean of the
estimates does not equal \(\Theta\) --- indeed, it is biased by the
shrinkage factor. However, the posterior mean is unbiased in the
Bayesian sense, that it is exactly the expected value we would assign to
\(\Theta\) having observed the data (in light of the assumed prior
distribution). See \citet{sennCommentaryTransposedConditionals2008} for
a discussion of bias in this ``forward'' and ``backwards'' sense.
Finally, \(x=1\) in the paradox, and so the posteriors agree between the
two investigators when B continues past the interim.

\subsubsection{Example}\label{example}

We have simulated a small study in accordance with the formulation of
the paradox with \(n =\) 5 and true mean \(\theta_0=\) 2 and \(\sigma=\)
2. The study used \(\psi=\) 1 and observed \(\dotbar{y}^{(1)} =\) 0.77
and investigator B did thus not stop at interim. The final mean was
\(\dotbar{y} =\) 0.88. The data is shown in
Table~\ref{tbl-simData-table}.

Under the frequentist paradigm, the maximum likelihood estimator
\(\hat{\theta}_{\text{MLE}}\) is the final mean \(\dotbar{y} =\) 0.88.
In light of the bias in this estimator and the fact that the study did
not stop at the interim, investigator \(B\) might opt to apply a bias
correction and use the estimator,
\begin{equation}\phantomsection\label{eq-q1-freq-bias-correction}{
\hat{\theta}_{\text{BC}} = \hat{\theta}_{\text{MLE}} + \frac{\sigma}{2\sqrt{n}} \frac{\phi \left(\frac{\sqrt{n}}{\sigma} \left[\psi - \hat{\theta}_{\text{MLE}} \right]\right)}{\Phi \left(\frac{\sqrt{n}}{\sigma} \left[\psi - \hat{\theta}_{\text{MLE}} \right]\right)} .
}\end{equation} The bias-correction amounts to adding the conditional
bias from (\ref{eq-conditional-bias-x1}) and plugging in the maximum
likelihood estimate for \(\Theta\). Doing so yields
\(\hat{\theta}_{\text{BC}} =\) 1.2, which in this case brings the
estimate somewhat closer to the true mean.

Under the Bayesian paradigm, both investigators pose the prior
\(N(1, 4)\) (i.e.~\(\mu=\) 1 and \(\tau=\) 2).
Figure~\ref{fig-posteriors-q1} depicts the posterior distribution for
the two investigators with the theoretical posterior density from
(\ref{eq-posterior-q1}) overlayed. The Stan code used to draw from the
posterior distribution is given in Appendix \ref{sec-stan-program}.

\begin{longtable}[]{@{}rrr@{}}

\caption{\label{tbl-simData-table}Simulated data with \(n=5\),
\(\theta_0 = 2\), \(\sigma = 2\), and \(\Psi=1\). The column \texttt{y1}
is the data collected before the interim analysis, while \texttt{y2} is
the data collected after. The study did not stop at interim
(\texttt{x}=1).}

\tabularnewline

\toprule\noalign{}
y1 & y2 & x \\
\midrule\noalign{}
\endhead
\bottomrule\noalign{}
\endlastfoot
-0.0716906 & 3.509635 & 1 \\
1.5528526 & -2.461906 & 1 \\
1.8782791 & -1.299701 & 1 \\
0.2941379 & 2.021037 & 1 \\
0.2096947 & 3.169979 & 1 \\

\end{longtable}

\begin{figure}

\centering{

\pandocbounded{\includegraphics[keepaspectratio]{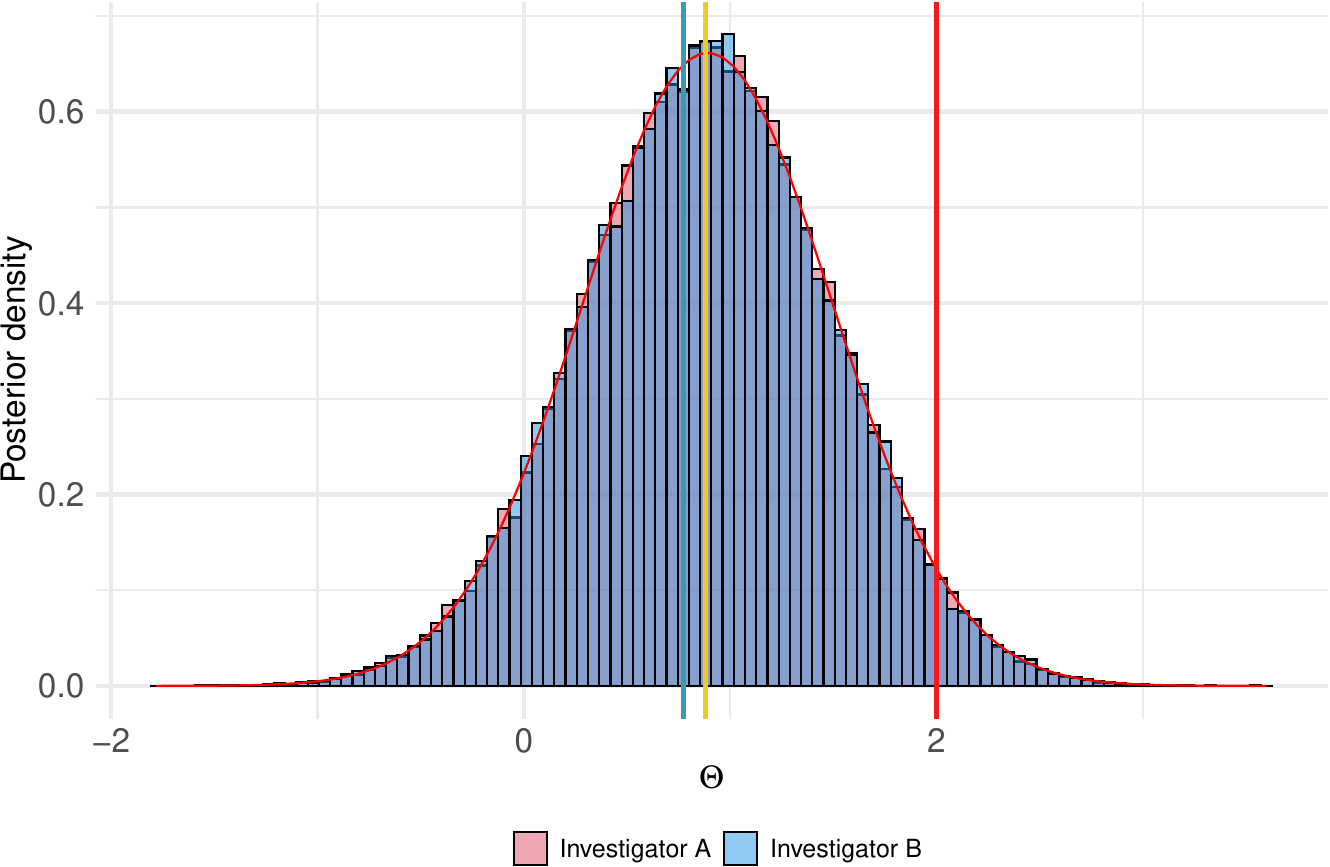}}

}

\caption{\label{fig-posteriors-q1}Posterior distribution for the two
investigators \(A\) and \(B\). The theoretical posterior from
(\ref{eq-posterior-q1}) is overlayed. Three vertical lines denote
\(\dotbar{y}^{(1)}\) (blue), \(\dotbar{y}\) (yellow), and \(\theta_0\)
(red).}

\end{figure}%

We see that the posterior distribution is the same for the two
investigators as in (\ref{eq-posterior-q1}) despite the investigators'
differing intentions regarding the interim analysis. Using the posterior
mean (coinciding here with the posterior mode) as point estimate for
\(\Theta\), the two investigators both estimate the mean as 0.8911.

In other words, investigator B will in the Bayesian paradigm apparently
not need to compensate in any way for their original intentions to stop
the trial in the presence of a stronger response (contrary to the
frequentist paradigm). This discrepancy between the frequentist and
Bayesian interpretation is the reason that the example may be described
as paradoxical.

\subsection{Question 2 and an alternative
interpretation}\label{sec-q2-and-alternative-interpretation}

Question 2 is, to our knowledge, put forth first by
\citet{sennStatisticalIssuesDrug2008} who identifies it as a difficulty
in posing the paradox from a Bayesian standpoint: If one accepts that a
Bayesian incorporates all their prior knowledge into the trial, one
would expect two Bayesians with the same priors to perform the same
trial. It would by contradiction follow from the different behaviour of
investigators \(A\) and \(B\) that the two should harbour different
priors for the parameters. This is in conflict with the formulation of
the paradox which necessitates identical priors. We analyse and expand
on this point in the following by expanding the setup to model how the
two investigators arrive at different intentions on the original design.

Let \(P = A, B\) be a random variable indicating the investigator.
Consider the DAG in Figure~\ref{fig-simple-case-expand}, which is an
expansion of Figure~\ref{fig-simple-case}. The DAG treats the design
parameter \(\Psi\) as a statistical parameter. This model setup posits
that the decision to stop at an interim (as represented by \(X\)) is a
combination of two sources, knowledge (as represented by \(\Psi\)) and
personality (as represented by \(P\)). In the setup, it is assumed that
a person (not necessarily investigator A or B) will employ the threshold
\(\Psi\), and while this is indeed the threshold employed by
investigator B, the threshold is by investigator A overruled in favour
of their personal preference to continue the trial. The broader decision
to represent \(\Psi\) as a random parameter may also require further
elaboration. Superficially, \(\Psi\) is set at the design phase of the
study and one may wonder it what sense it is random. Consider the
calibration representation of the posterior: We create a universe of
relevant trials by sampling the parameters from the prior distribution
and outcome variable from the likelihood function given the value of the
parameters. What is a reasonable assumption as to how these trials are
performed? One choice is to let all trials use the same \(\Psi\)
regardless of \(\Theta\). Another is to expand the universe of relevant
trials slightly and let each trial use different values of \(\Psi\).
Under this second option, is seems reasonable to allow the value of
\(\Psi\) to depend on \(\Theta\) so that each relevant trial may choose
their design parameter depending on the \(\Theta\) they have been
assigned. However, if \(\Psi\) is fixed, this corresponds to an
expression of uncertainty about the treatment effect yet complete
certainty about how the study should be designed. As we shall elaborate
on in the discussion. We may thus think of \(\Psi\) as an opportunity to
encode our uncertainty about the study design into the analysis. As a
related point, \citet{freedmanComparisonBayesianGroup1989} study
Bayesian stopping boundaries that depend on the precision of the prior
information on the treatment effect and show how these may, under
differing prior precisions, mimic prevalent frequentist boundaries.

\begin{figure}

\centering{

\pandocbounded{\includegraphics[keepaspectratio]{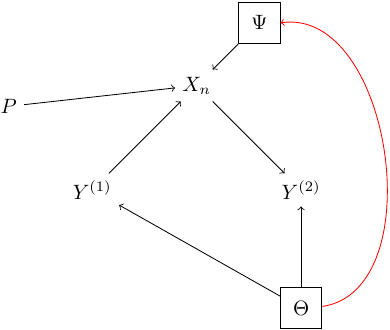}}

}

\caption{\label{fig-simple-case-expand}Expanded simple case DAG.}

\end{figure}%

Some thought is required in formulating the paradox and in particular
when defining the likelihood. We define the likelihood as conventionally
done: as the joint density of the observed data conditional on the
parameters. However: What constitutes the data? By whom is it observed?
What are the parameters? We approach the answers to these questions
inductively: The perhaps most obvious choice would be to consider the
density of the conditional distribution \({(Y,X,P) \mid \Theta, \Psi}\)
thus including all variables and parameters from the DAG in
Figure~\ref{fig-simple-case-expand}. But by whom is this data observed?
This is the data observed by us as formulators of the paradox. The two
investigators in the paradox do not observe \(P\). If we accept this
argument, we might propose to base our likelihood instead on
\({(Y,X) \mid P, \Theta, \Psi}\). Note that the corresponding prior is
\((\Theta, \Psi) \mid P\) which does not depend on \(P\) in the DAG in
Figure~\ref{fig-simple-case-expand}. This concedes that the two
investigators may have the same prior and is thus in line with the
formulation of the paradox from a Bayesian perspective. For investigator
A, the study is still a fixed design as in
Figure~\ref{fig-fixed-design}, since for investigator A, the
distribution of \(X\) is degenerate (\(X\) is always one) so that the
significance of the red arrow is immaterial.

We derive the posterior distribution of the treatment effect given the
investigator, \(f_{\Theta \mid X, Y, P}\), and the paradox may be
summarised as the question of whether this posterior depends on the
investigator \(P\). This involves marginalising over the design
parameter \(\Psi\) (e.g.~\citet{liseoEliminationNuisanceParameters2005})
and we obtain (details in Appendix \ref{sec-derivations-q2}),

\begin{equation}\phantomsection\label{eq-q2-posterior}{
\begin{aligned}
  f_{\Theta \mid X, Y, P}
  &=
    \frac{\left(\int  f_{X \mid Y^{(1)}, \Psi, P}  f_{\Psi \mid \Theta} \; d \psi \right) f_{Y^{(2)} \mid X, \Theta} f_{Y^{(1)} \mid \Theta} f_{\Theta}}{\int \left(\int f_{(X,Y) \mid \Theta, \Psi, P} f_{\Psi \mid \Theta, P} \: d \psi\right) f_{\Theta \mid P} \: d \theta }
.
\end{aligned}
}\end{equation} We note that the posterior, up to a normalising
constant, factors into two parts, where the second factor outside the
integral does not depend on the investigator \(P\). The integrand (the
design likelihood) does however depend on \(P\) and the integration is
over the conditional distribution of the design parameter \(\Psi\) given
\(\Theta\). It is seen that the crux of the argument is the red arrow in
Figure~\ref{fig-simple-case-expand} -- if \(\Psi\) is taken to be
independent of \(\Theta\) then the integral no longer depends on the
parameter of interest and may be absorbed into the normalising constant
so that the posterior no longer depends on the investigator. If,
conversely, there is dependence between the treatment and design
parameter, then the integral depends on both \(\Theta\) and the
investigator.

To illustrate further the difference between the posteriors of the two
investigators in the setup corresponding to
Figure~\ref{fig-simple-case-expand}, we continue our calculations from
the simple case considered in
Section~\ref{sec-q1-and-conventional-analysis} above. The sampling model
is \(Y_i \mid \Theta, X \sim N(\Theta, \sigma^2)\) for known
\(\sigma > 0\) and \(i = 1, \ldots, (1+X)n\). We continue our use of the
prior \(\Theta \sim N(\mu, \tau^2)\) and define the prior on \(\Psi\)
from a linear model in \(\Theta\) by setting, \[
\Psi = a + b \cdot \Theta + \epsilon , \quad \epsilon \sim N (0, \omega^2) \quad \epsilon \indep \Theta ,
\] We first note that for investigator A, the integral in the nominator
of (\ref{eq-q2-posterior}) is a constant in \(\theta\), and thus the
derivations of the posterior reduce to those performed above in
Section~\ref{sec-q1-and-conventional-analysis}. For investigator \(B\),
under the chosen model setup, the posterior admits an analytically
closed form (details in Appendix \ref{sec-derivations-q2}), the
posterior density being given by,
\begin{equation}\phantomsection\label{eq-posterior-q2-B}{
\begin{aligned}
&f_{\Theta \mid X, Y, P} (\theta \mid x, y, B) = \\
&\quad\ExpressionPosteriorDensity
\end{aligned}
}\end{equation}

This posterior distribution for \(\Theta\) is calibrated to a universe
of trials where \(\Psi\) is allowed to vary and is dependent on
\(\Theta\) as expressed through the hyperparameter \(b\). Note that for
\(b=0\) the posterior simplifies to the distribution in
(\ref{eq-posterior-q1}), so that the density in this case agrees between
the two investigators when they both continue beyond the interim
analysis (\(x=1\)). Moreover, the posterior mean may be derived as,
\begin{equation}\phantomsection\label{eq-q2-posterior-mean}{
\cexpect{\Theta}{X = x, Y = y, P = B}
=
\ExpressionPosteriorMean
}\end{equation} We see that the posterior mean agrees with the posterior
mean \(\mu_{1+x}\) of investigator A plus a second term. When \(b=0\)
this second term disappears, so that the term may be said to stem from
the design and it may be interpreted as a bias correction as we
elaborate on in the subsequent section.

\subsubsection{Example (continued)}\label{example-continued}

Figure~\ref{fig-posteriors-q2} depicts the posterior distributions of
investigator \(A\) and \(B\) having observed the same data as above
given in Table~\ref{tbl-simData-table}, but now under the expanded
formulation of the paradox. The hyperparameters of the priors are shown
in Table~\ref{tbl-hyper-expanded-paradox}.

\begin{longtable}[]{@{}rrrrrr@{}}

\caption{\label{tbl-hyper-expanded-paradox}Hyperparameters in the
expanded paradox.}

\tabularnewline

\toprule\noalign{}
\(\sigma\) & \(\mu\) & \(\tau\) & \(a\) & \(b\) & \(\omega\) \\
\midrule\noalign{}
\endhead
\bottomrule\noalign{}
\endlastfoot
2 & 1 & 2 & -0.5 & 1 & 0.1 \\

\end{longtable}

\begin{figure}

\centering{

\pandocbounded{\includegraphics[keepaspectratio]{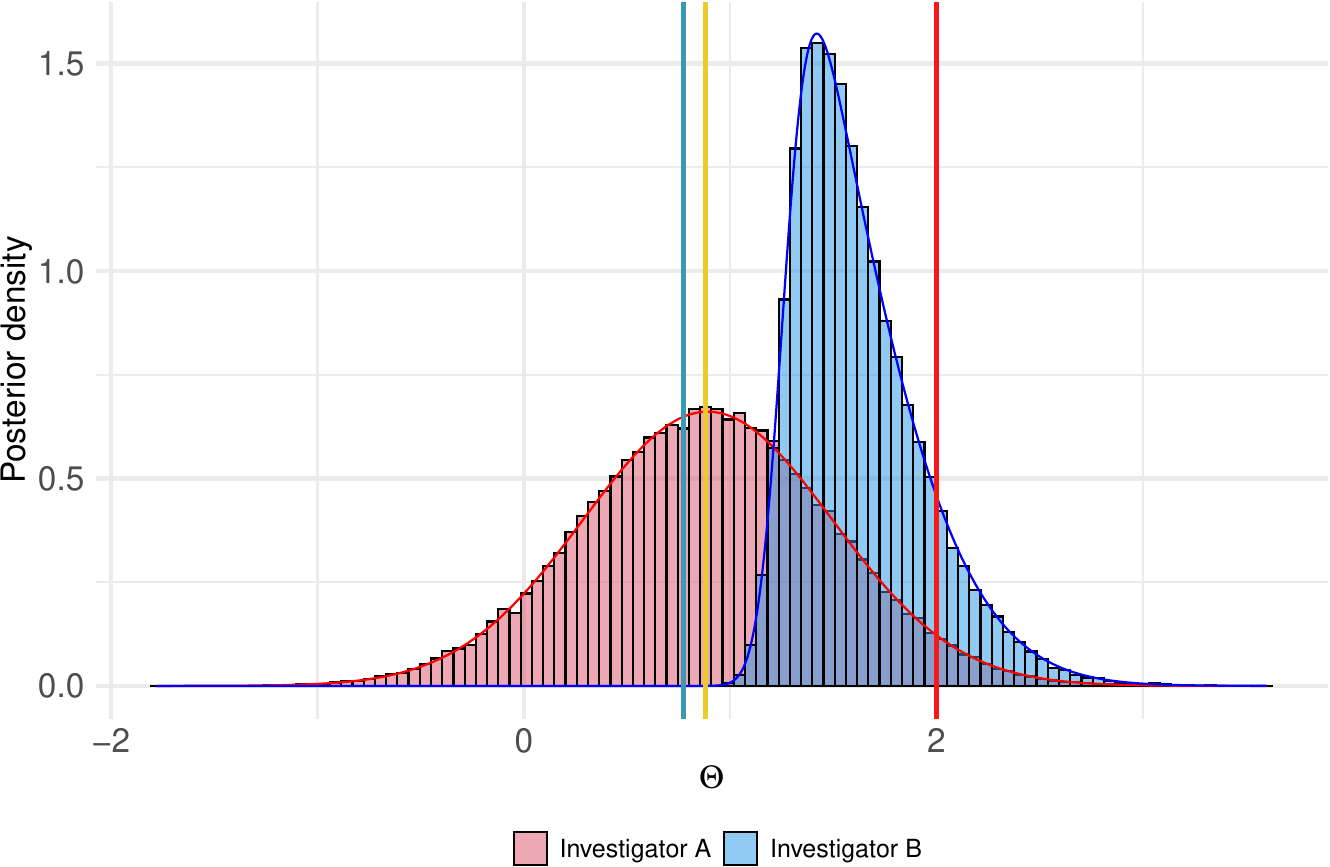}}

}

\caption{\label{fig-posteriors-q2}Posterior distribution for the two
investigators \(A\) and \(B\) under the expanded formulation of the
paradox. The superimposed red curve shows the theoretical posterior
(\ref{eq-posterior-q2-A}) for investigator \(A\) while the blue curve is
the theoretical posterior (\ref{eq-posterior-q2-B}) of investigator
\(B\). Three vertical lines denote \(\dotbar{y}^{(1)}\) (blue),
\(\dotbar{y}\) (yellow), and \(\theta_0\) (red).}

\end{figure}%

We see that the posterior for investigator A coincides with that in
Figure~\ref{fig-posteriors-q1} as expected. For investigator B the
posterior is shifted conciderably upwards and is also noticeably skewed
to the right. Posterior means and modes are shown in
Table~\ref{tbl-posterior-means-q2}. For both investigators, the
theoretical and empirical posterior means are the same (up to Monte
Carlo error), and the mean coincides with the mode for Investigator A,
as expected from the normal distribution, while the mode is slightly
smaller than the mean for investigator B as expected from the
right-skewed posterior.

\begin{longtable}[]{@{}lrr@{}}

\caption{\label{tbl-posterior-means-q2}Table of posterior mean and mode
estimates for the two investigators calculated empirically from the
posterior draws in Figure~\ref{fig-posteriors-q2}, or, for the mean,
analytically using formula (\ref{eq-q2-posterior-mean}).}

\tabularnewline

\toprule\noalign{}
& Investigator A & Investigator B \\
\midrule\noalign{}
\endhead
\bottomrule\noalign{}
\endlastfoot
Mean of posterior draws & 0.8887 & 1.6211 \\
Mode of posterior draws & 0.9270 & 1.3845 \\
Theoretical posterior mean & 0.8911 & 1.6247 \\

\end{longtable}

\subsection{Prior specifications and estimator
bias}\label{prior-specifications-and-estimator-bias}

The significance of the dependence between \(\Theta\) and \(\Psi\) in
the prior distribution may be framed in terms of estimation bias in the
frequentist sense. Following the line of inquiry in
\citet{dawidSelectionParadoxesBayesian1994}, we analyse the magnitude of
the frequentist bias as a function of the parameters and compare these
parameter scenarios to the choice of priors.

From (\ref{eq-bias-freq}) we see that the estimation bias depends on
\((\Theta, \Psi)\) only through \((\Psi - \Theta)^2\) and is
exponentially decreasing in this term, so that a large quadratic
difference in the two parameters corresponds to a small bias. The bias
is maximised when \(\Psi = \Theta\). The conditional biases given
stopping (\(X=0\)) and continuation (\(X=1\)) are on the other hand
maximised when \(\Psi \gg \Theta\) and \(\Psi \ll \Theta\),
respectively.

Note that the function \(z \mapsto \phi(z)/\Phi((-1)^{1-x} z)\) is
strictly increasing when \(x = 0\) and strictly decreasing for \(x=1\).
Consider the factor from the second term in the posterior mean
(\ref{eq-q2-posterior-mean}),
\begin{equation}\phantomsection\label{eq-q2-bias-correction-term}{
    (-1)^{1-x} b
    \frac{
    \sigma^2_{1 + x}
    }{\sqrt{\omega^2 + \sigma_{1+x}^2 b^2}}
    \frac{
    \phi
    \left(
    \frac{a + b \mu_{1 + x}  - \bar{y}_{\boldsymbol{\cdot}}^{(1)} }{\sqrt{\omega^2 + \sigma_{1+x}^2 b^2}}
    \right)
    }
    {
    \Phi
    \left(
    (-1)^{1-x}
    \frac{a + b \mu_{1 + x}  - \bar{y}_{\boldsymbol{\cdot}}^{(1)} }{\sqrt{\omega^2 + \sigma_{1+x}^2 b^2}}
    \right)
    }
}\end{equation} The term closely resemples the truncation factors in the
conditional biases (\ref{eq-conditional-bias-x0}) and
(\ref{eq-conditional-bias-x1}). We note that the sign of this extra term
depends on \(x\) similarly to the observation that the sign of the
conditional biases in (\ref{eq-conditional-bias-x0}) and
(\ref{eq-conditional-bias-x1}) are opposite and indeed the sign of the
term corrects in the opposite direction of the bias (if \(b > 0\)). Note
the significance of the first period mean \(\dotbar{y}^{(1)}\) in the
formula. This estimator is unbiased for \(\Theta\) since it is based on
data only up to the adaptation. Based on these observations, the
expression in (\ref{eq-q2-bias-correction-term}) may be perceived as an
estimator for the conditional bias (given a set of hyperparameters).

We may tune the hyperparameters of the \(\Psi\) prior to emulate the
situation where the bias is small, i.e.~where \(\Psi\) is numerically
very different from \(\Theta\) (so that the quadratic difference will be
large), by taking \(b = 0\) or taking \(\omega\) to be large
(i.e.~\(\omega \uparrow \infty\)). Inspecting the posterior mean we see
that it will coincide with that of investigator \(A\), since the second
term disappears. Thus, if our prior specification indicates that the
bias is small, the posterior distribution does not include the bias
correction term (\ref{eq-q2-bias-correction-term}).

When \(\omega\) is finite and \(b\) is non-zero a large (small) \(a\)
corresponds to a prior belief that \(\Psi \gg \Theta\)
(\(\Psi \ll \Theta\)) signifying a predisposition towards continuing
(stopping) the trial, and given this prior specification, the posterior
reacts by imposing a small (large) bias correction whenever the trial
stops (continues). These scenarios where \(\Psi \gg \Theta\) or
\(\Psi \ll \Theta\) correspond to those where the conditional biases are
maximised.

We may emulate the frequentist bias correction in
(\ref{eq-q2-bias-correction-term}) through the choice of
hyperparameters. We focus on the scenario with continuation (\(x=1\)) as
in our example. In order to emulate the frequentist paradigm we let
\(\tau\) tend to infinity to obtain a flat prior on \(\Theta\) which in
turn means that \(\sigma_{1+x}\) will tend to \(\sigma/\sqrt{2n}\) and
\(\mu_{1+x}\) to \(\dotbar{y}\). Take \(b>0\) and set
\(\omega = b \cdot \sigma/\sqrt{2n}\) and \(a = b \cdot \psi\), where
\(\psi\) is the stopping threshold used in the study. Then the
correction term from (\ref{eq-q2-bias-correction-term}) is, \[
    \frac{\sigma}{2 \sqrt{n}}
    \frac{
    \phi
    \left(
    \frac{\sqrt{n}}{\sigma}
    \left[
    \psi - (\frac{1}{b} \dotbar{y}^{(1)} - \dotbar{y} )
    \right]
    \right)
    }
    {
    \Phi
    \left(
    \frac{\sqrt{n}}{\sigma}
    \left[
    \psi - (\frac{1}{b} \dotbar{y}^{(1)} - \dotbar{y} )
    \right]
    \right)
    }
    .
\] This form of the correction term highlights the obvious similarities
but also a difference in the frequentist and Bayesian bias correction:
In the frequentist bias correction in
(\ref{eq-q1-freq-bias-correction}), the size of the correction depends
on the observed mean \(\dotbar{y}\). In the Bayesian correction, the
size of the correction may be made to depend on the difference between
the first period mean \(\dotbar{y}^{(1)}\) and the overall mean
\(\dotbar{y}\) (or equivalently comparing the first and second period
means) depending on the choice of the parameter \(b\) which acts as a
deflation parameter for the first period mean. The largest bias
correction is applied in the scenario where \(\psi\) is small (so that
we were quite likely to stop the study) and the first period mean is
large (but still bounded upwards by \(\psi\) or else we would not have
continued) and the second period mean is very small.

Finally if the overall and first period means have the same sign we
could take
\(b=\dotbar{y}^{(1)}/(\dotbar{y}^{(1)} + \dotbar{y}^{(2)}) = \dotbar{y}^{(1)}/2 \dotbar{y} > 0\)
to be the relative difference between the two and the Bayesian
correction term would coincide with the frequentist correction term.
Naturally, such a data-dependent choice of prior parameter would be
prohibited under a traditional Bayesian paradigm but it is of
theoretical interest not least because it shows how the frequentist
bias-adjusted estimator may be derived as an empirical Bayes estimator
(where hyperparameters are estimated from data). Thus, the posterior
mean in the present setup may be said to generalise the frequentist bias
corrected estimator when cast in this empirical Bayes framework.

\subsubsection{Example (continued)}\label{example-continued-1}

In the analysis of the data in the example above, assuming both
investigators give the posterior empirical mean as a point estimate,
investigator A would estimate \(\Theta\) as 0.89. Investigator B would
on the other hand, give the somewhat larger estimate 1.62, in
appreciation of the fact that the study did not stop at interim, so that
there is a risk of underestimating \(\theta\). In this specific example,
the bias correction moves the estimate closer to the true
\(\theta_0 = 2\).

\section{Discussion}\label{discussion}

It is not surprising that the investigators in the extended formulation
of the paradox should arrive at different posterior distributions since
as may be verified from Figure~\ref{fig-simple-case-expand} the design
is not well-behaved in the sense of Proposition~\ref{prp-wb-likel}: when
viewed marginal to \(\Psi\) the setup simply corresponds to drawing an
arrow from \(\Theta\) to \(X\) in Figure~\ref{fig-simple-case}. The
perhaps more poignant idea here is that the deviation from a
well-behaved design is not caused by the design \emph{per se} but rather
by the specific conceptualisation and modelling of the study through the
prior: Do we view \(\Psi\) as a parameter and, if so, does it depend on
\(\Theta\)?

The expanded DAG in Figure~\ref{fig-simple-case-expand} may be viewed as
a hierarchical model introducing dependence between \(\Theta\) and the
design variable \(X\) through \(\Psi\). The hierarchical construction is
also used in \citet{sennNoteConcerningSelection2008} (for a continuous
outcome) and \citet{mandelSelectionBiasConflict2009} (binary outcome).
The difference between these papers and the present is that in these two
papers (as in \citet{dawidSelectionParadoxesBayesian1994}) the focus is
on multiple parameters among which selection occurs. Similarly to our
findings, these papers also conclude that Bayesian inference is affected
by the selection if there is dependence between the parameters.
\citet{harvilleBayesianInferenceUnaffected2022} proposes to model
selection more directly as a modification of the likelihood component of
the posterior, but also discusses how this may be subsumed under the
prior. Similarly, note that the term in arising from our choice of prior
(the term in (\ref{eq-q2-integral-investigator-B}) in Appendix
\ref{sec-derivations-q2}) could also be moved to the likelihood part so
that the data correspond to those arising from a sort of truncated
normal distribution. However, we find the motivation of the term through
the prior more intuitive as a differing likelihood would imply a
different choice of sampling model and thus would not be a useful
setting for analysing the paradox where we have constrained the sampling
models to be the same.

The proposed hierarchical model that allows for a dependence between
\(\Theta\) and the design parameter \(\Psi\) is easy to implement as
shown in the Stan code in Appendix \ref{sec-stan-program}. Under a
simple sampling model and prior distribution we could derive the
posterior distribution in equation (\ref{eq-posterior-q2-B}). The
posterior is unlikely to be analytically tractable in more complicated
models but the general form of the posterior in equation
(\ref{eq-q2-posterior}) can be sampled using a Markov chain Monte Carlo
procedure such as Stan. We leave the potential of this framework for
constructing bias-adjusted estimators as a line of future research.

It must also be stressed that the expanded model in
Figure~\ref{fig-simple-case-expand} is only one of many extensions and
possible motivations for a dependence between the parameter of interest
\(\Theta\) and the design. As an alternative, one could for example
posit that the two are related through a common ancestor in the DAG
called ``clinical knowledge''. A related point is what one considers
(relevant) parameters in the study. In principle, as pointed out in
\citet{sennStatisticalIssuesDrug2008}, a prior distribution could be
attached to any auxiliary parameter for the design, e.g.~the sample
size, the expected gain and loss in terms of clinical efficacy and
health provider cost etc. The current practise however is to focus on
the prior for the parameter of interest, usually the treatment effect,
and Senn identifies this practise as stemming from early influential
papers such as \citet{spiegelhalterBayesianApproachesRandomized1994}
that argued for the use of Bayesian methods for clinical trials and
where this pragmatic approach is taken. Based on our investigations
above, a crucial point seems to be that this pragmatic approach may be
warranted in a fixed design, but does not necessarily translate to an
adaptive design: The arrow from \(\Theta\) to \(\Psi\) in
Figure~\ref{fig-simple-case-expand} does not matter to Investigator
\(A\) who is running a fixed design, and as such, the arrow could simply
be omitted. But the arrow matters tremendously for the interpretation of
the adaptive design of Investigator B, as it will determine whether the
interim analysis can be ignored or not when deriving the posterior
distribution.

We can motivate our results from the calibration representation of the
posterior as a sampling problem: For both investigators we construct a
universe of relevant trials by first sampling the parameters
\((\Psi, \Theta)\) and then the data conditional on the parameters for a
high number of repetitions. Suppose that we have observed data \(x=1\)
and \(y\). We recognise our observed data \((x,y)\) (since \(y\) is
continuous this is up to rounding) among some of the repetitions (for
investigator \(B\) we only look at repetitions where \(X=1\)) and wish
to form the posterior distribution from the corresponding \(\Theta\)'s.
Can we combine the samples from investigator A and B? When the
parameters are sampled independently from each other, the answer is an
affirmative: \(\Theta\) follows the same distribution
(\ref{eq-posterior-q1}) for the two investigators. However, when there
is dependence between the parameters, the samples for investigator \(B\)
are no longer representative for those for investigator \(A\). This may
be seen as an instance of collider stratification bias
(\citet{greenlandCausalDiagramsEpidemiologic1999}) when conditioning on
the data \((x,y)\) among the samples: When the red arrow is present in
Figure~\ref{fig-simple-case-expand} the conditioning opens a path
between the investigator and \(\Theta\). The difference in the
\(\Theta\) distribution is due to our conditioning on continuation
\(x=1\), which changes the distribution of \(\Psi\) (continuation is
more likely for larger \(\Psi\)'s) and thus (when they are dependent)
the distribution of \(\Theta\). For investigator \(A\) the distribution
is not changed, because the distribution of \(\Psi\) is immaterial when
always continuing, but for investigator \(B\) we expect a higher number
of large \(\Theta\)'s in our sample, since (when \(b>0\)) the parameters
are positively correlated which allows a large \(\Theta\) more often,
since it is likely to be accompanied by a large \(\Psi\). Calibration is
in this sense a fairly weak property: The two investigators will be
calibrated with themselves, but the two posterior distributions will not
necessarily agree, so that we might say that they are not calibrated to
each other. This is opposite to the view of the calibration property
taken in Frank Harrell's blog (\citet{harrellblog}).

Our final remarks concern the choice of priors. Arguably, there is no
such thing as an objective prior distribution. To specify a prior
independently of the investigator and analyst, an interesting line of
approach elicits expert opinions about the treatment effect and uses
these to construct a prior distribution
(e.g.~\citet{parmarMonitoringLargeRandomised2001}). Often, several
priors are specified by obtaining different sets of opinions for example
to construct a sceptical and optimistic prior among experts who are
respectively unconvinced and excited about the line of treatment. The
resulting posteriors then represent interpretations of the trial data
among such experts of diverging opinions. In light of our derivations,
it would be interesting to expand such an elicitation to also obtain
information on design, i.e.~asking the clinician first about their
expectations concerning the treatment effect (\(\Theta\)) and then
something like \emph{Suppose that you are given the option to stop the
trial at halfway to the final sample size, at which observed treatment
effect do you think it would be sensible to stop the trial?}. Having the
expert's opinions on the pairs \((\Theta, \Psi)\) it would be
interesting to see if any correlation exists between the two ``in the
wild''.

It has been appreciated, so for example with Lindley's paradox (see for
example \citet{shaferLindleysParadox1982}), that a discrepancy between
frequentist and Bayesian inference can be due to non-informative priors.
If one accepts that a prior must be placed on the pair of treatment and
design parameter \((\Theta, \Psi)\), then specifying the two as
independent would arguably be a sort of non-informativeness in the joint
prior distribution and thus place our findings in this more general
frame. Our experience is that the location and scale parameters of a
prior received comparably more attention than does e.g.~the correlations
in the prior distribution. Another and perhaps somewhat aporetic view of
the paradox attributes it to the fact that the frequentist and Bayesian
paradigms simply address two different questions (as shown by the
orthogonal inferences in Figure~\ref{fig-twodim-inference}), see
\citet{emersonBayesianEvaluationGroup2007}.

The Bayesian interpretation of the adaptive design of investigator \(B\)
yields a posterior distribution that is invariant to the interim
analysis only under the specific prior choice of independent \(\Theta\)
and \(\Psi\). Therefore, this invariance property should not be
considered universal to a Bayesian analysis. It is simply a results of
the posterior distribution reacting to the prior information that the
design contains no information about the parameter of interest. This
also yields an opportunity to let the posterior incorporate the design,
which may provide comfort to those
(e.g.~\citet{dawidSelectionParadoxesBayesian1994}) who find the lack of
reaction in the posterior distribution to the changed design a
``weakness'' of Bayesian inference. The prior should be carefully
specified and not some default chosen e.g.~by software as also concluded
in \citet{dawidSelectionParadoxesBayesian1994}. We make the additional
important point that one must, before specifying a prior, decide on
which parameters to model, as in an adaptive design, in contrast to a
fixed design, the role of design parameters may crucially affect the
interpretation of the trial data.

\section{References}\label{references}

\renewcommand{\bibsection}{}
\bibliography{bibliography.bib}

\appendix

\section{Appendix: Derivations}\label{sec-derivations}

\subsection{\texorpdfstring{Proof of
Proposition~\ref{prp-wb-likel}}{Proof of Proposition~}}\label{sec-proof-prp-wb-likel}

The statement of the the proposition follows directly from sequential
conditioning on the likelihood function,

\[
\begin{aligned}
L(\theta) 
    &= 
    f_{(\mathbf{X}, \mathbf{Y}) \mid \Theta}  \\
    &= 
    \prod_{i = 1}^N f_{(X_i, Y_i) \mid \history{X_i}, \history{Y_i}, \Theta} \\
    &= 
    \prod_{i = 1}^N f_{Y_i \mid X_i, \history{X_i}, \history{Y_i}, \Theta} 
        f_{X_i \mid \history{X_i}, \history{Y_i}, \Theta} \\
    &\overset{\text{W-B}}{=} 
    \prod_{i = 1}^N f_{Y_i \mid X_i, \Theta} 
        f_{X_i \mid \history{X_i}, \history{Y_i}} \\
    &\propto
    \prod_{i = 1}^N f_{Y_i \mid X_i, \Theta} ,
\end{aligned}
\]

where we in the second-to-last line applied the assumption of a
well-behaved design.

\subsection{\texorpdfstring{Bias expression in
Section~\ref{sec-q1-and-conventional-analysis}}{Bias expression in Section~}}\label{sec-deriv-bias-expression}

We first derive the frequentist bias leading to equation
(\ref{eq-bias-freq}). Note that the estimator may be written as,

\[
\hat{\theta} (X, Y) = \frac{1}{(1 + X)n} \left\{ \sum_{i=1}^{n} Y_i + X \sum_{i=n+1}^{2n} Y_i \right\}
\]

and that,

\[
\begin{aligned}
  \cexpect{\hat{\theta} (X, Y)}{\Theta, \Psi} =
  \cexpect{\hat{\theta} (X, Y)}{\Theta, \Psi, X = 0} \left[1 - \pi (\Theta, \Psi) \right]
  + \cexpect{\hat{\theta} (X, Y)}{\Theta, \Psi, X = 1} \pi (\Theta, \Psi)
\end{aligned}
\]

with \(\pi (\Theta, \Psi) = \cproba{X = 1}{\Theta, \Psi}\). Using
Lemma~\ref{lem-truncated-normal} below, we find that,

\[
\begin{aligned}
  \cexpect{\hat{\theta} (X, Y)}{\Theta, \Psi, X = 0} 
  &=
  \cexpect{\frac{1}{n} \sum_{i=1}^{n} Y_i}{\Theta, \Psi, \frac{1}{n} \sum_{i=1}^{n} Y_i > \Psi}  \\
  &=
    \Theta + \frac{\sigma}{\sqrt{n}} \frac{\phi \left(\frac{\sqrt{n}}{\sigma} \left[\Psi - \Theta \right]\right)}{1 - \Phi \left(\frac{\sqrt{n}}{\sigma} \left[\Psi - \Theta \right]\right)} ,
\end{aligned}
\] by using that
\(\frac{1}{n} \sum_{i=1}^{n} Y_i \mid \Theta, \Psi \sim N(\Theta, \frac{\sigma^2}{n})\).
Similarly, we find that,

\[
\begin{aligned}
  \cexpect{\hat{\theta} (X, Y)}{\Theta, \Psi, X = 1} 
  &=
    \Theta - \frac{\sigma}{2\sqrt{n}} \frac{\phi \left(\frac{\sqrt{n}}{\sigma} \left[\Psi - \Theta \right]\right)}{\Phi \left(\frac{\sqrt{n}}{\sigma} \left[\Psi - \Theta \right]\right)} ,
\end{aligned}
\]

and

\[
\pi(\Theta, \Psi) = \Phi \left(\frac{\sqrt{n}}{\sigma} \left[\Psi - \Theta \right]\right) .
\]

So plugging in yields,

\[
\begin{aligned}
  \cexpect{\hat{\theta} (X, Y)}{\Theta, \Psi} = \Theta + \frac{\sigma}{2\sqrt{n}} \phi \left(\frac{\sqrt{n}}{\sigma} \left[\Psi - \Theta \right]\right) .
\end{aligned}
\]

\begin{lemma}[]\protect\hypertarget{lem-truncated-normal}{}\label{lem-truncated-normal}

Let \(W \sim N(\mu, \sigma^2)\). Then

\[
\cexpect{W}{a \leq W \leq b} = \mu + \sigma\frac{\phi\left(\frac{a - \mu}{\sigma} \right) - \phi\left(\frac{b - \mu}{\sigma} \right)}{\Phi\left(\frac{b - \mu}{\sigma} \right) - \Phi\left(\frac{a - \mu}{\sigma} \right)}
\]

\end{lemma}

\begin{proof}
Let Z be standard normal. Then

\[
\begin{aligned}
  \cexpect{Z}{a \leq Z \leq b}
  &=
    \frac{1}{\proba{a\leq Z \leq b}} \int_a^b z \phi(z) \: dz \\
  &=
    \frac{1}{\Phi(b) - \Phi(a)} \int_a^b -  \phi'(z) \: dz \\
  &=
    \frac{\phi(a) - \phi(b)}{\Phi(b) - \Phi(a)} ,
\end{aligned}
\] and the result follows by \(W \sim \mu + \sigma Z\).
\end{proof}

\subsection{\texorpdfstring{Posterior in
Section~\ref{sec-q1-and-conventional-analysis}}{Posterior in Section~}}\label{sec-posterior-in-q1}

The argument is standard, cf.~for example \citet{gelman2013bayesian}.,
but is included here for the sake of completeness. Note that the
posterior, \[
f_{\Theta \mid X, Y} \propto f_{Y^{(2)} \mid X, \Theta} f_{Y^{(1)} \mid \Theta} f_{\Theta}
\] so that denoting by \(c\) any constant in \(\theta\), we have that,
\begin{equation}\phantomsection\label{eq-posterior-q1-calc}{
\begin{aligned}
    \log f_{\Theta \mid X, Y} 
    &=
    - \frac{1}{2 \sigma^2} \sum_{i=1}^{(1+x)n} (y_i - \theta)^2 - \frac{1}{2 \tau^2} (\theta - \mu)^2 + c
    \\
    &=
    - \frac{1}{2 \sigma^2} \sum_{i=1}^{(1+x)n}
    \left\{
    \theta^2 \left(1 + \frac{\sigma^2}{\tau^2 (1+x)n} \right) 
    - 2 \theta \left( y_i + \mu \frac{\sigma^2}{\tau^2 (1+x)n} \right) 
    \right\} + c
    \\
    &=
    - \frac{1}{2 \sigma^2}
    \left\{
    \theta^2 \frac{\tau^2 (1+x)n + \sigma^2}{\tau^2} 
    - 2 \theta \left( y_{\boldsymbol{\cdot}} + \mu \frac{\sigma^2}{\tau^2} \right) \sqrt{\frac{\tau^2 (1+x)n + \sigma^2}{\tau^2}} \sqrt{\frac{\tau^2}{\tau^2 (1+x)n + \sigma^2}}
    \right\} + c
    \\
    &=
    - \frac{1}{2 \sigma^2}
    \left(
    \theta \sqrt{\frac{\tau^2 (1+x)n + \sigma^2}{\tau^2}} -
    \left( y_{\boldsymbol{\cdot}} + \mu \frac{\sigma^2}{\tau^2} \right) \sqrt{\frac{\tau^2}{\tau^2 (1+x)n + \sigma^2}}
    \right)^2 + c
    \\
    &=
    - \frac{1}{2 \frac{\sigma^2 \tau^2}{\tau^2 (1+x)n + \sigma^2}}
    \left(
    \theta -
    \left( y_{\boldsymbol{\cdot}} + \mu \frac{\sigma^2}{\tau^2} \right) \frac{\tau^2}{\tau^2 (1+x)n + \sigma^2}
    \right)^2 + c
\end{aligned}
}\end{equation} and we conclude that the posterior distribution is, \[
N
\left( 
 y_{\boldsymbol{\cdot}} \frac{\tau^2}{\tau^2 (1+x)n + \sigma^2} + \mu \frac{\sigma^2}{\tau^2 (1+x)n + \sigma^2} ,
 \frac{\sigma^2 \tau^2}{\tau^2 (1+x)n + \sigma^2}
\right)
\] or \[
N
\left( 
 \bar{y}_{\boldsymbol{\cdot}} \frac{\tau^2}{\tau^2 + \frac{\sigma^2}{(1+x)n}} + \mu \frac{\frac{\sigma^2}{(1+x)n}}{\tau^2 + \frac{\sigma^2}{(1+x)n}} ,
 \frac{\frac{\sigma^2}{(1+x)n} \tau^2}{\tau^2 + \frac{\sigma^2}{(1+x)n}}
\right)
\]

\subsection{\texorpdfstring{Posterior in
Section~\ref{sec-q2-and-alternative-interpretation}}{Posterior in Section~}}\label{sec-derivations-q2}

Consider the case where \[
Y_i \mid \Theta \overset{\text{iid}}{\sim} N(\Theta, \sigma^2), \quad \text{ for } i = 1,\ldots, n ,
\] and \[
Y_i \mid \Theta, X = 1 \overset{\text{iid}}{\sim} N(\Theta, \sigma^2), \quad \text{ for } i = n + 1,\ldots, 2n .
\]

The treatment effect prior is specified as, \[
\Theta \sim N(\mu, \tau^2) ,
\] while the design prior is given by, \[
\Psi = a + b \cdot \Theta + \epsilon, \quad \epsilon \sim N (0, \omega^2) ,  \quad \epsilon \indep \Theta .
\]

To derive the posterior, we must marginalise over the design parameter
\(\Psi\). Thus,

\[
\begin{aligned}
  f_{\Theta \mid X, Y, P}
  &=
    \frac{\left(\int f_{(X,Y) \mid \Theta, \Psi, P} f_{\Psi \mid \Theta, P} \; d \psi \right) f_{\Theta \mid P}}{\int \left(\int f_{(X,Y) \mid \Theta, \Psi, P} f_{\Psi \mid \Theta, P} \: d \psi\right) f_{\Theta \mid P} \: d \theta }
\end{aligned}
\]

We notice from the DAG in Figure~\ref{fig-simple-case-expand} that the
integrand (the \(\Psi\)-conditional likelihood), may be decomposed as,

\[
\begin{aligned}
  f_{(X,Y) \mid \Theta, \Psi, P}
  &=
    f_{Y^{(2)} \mid Y^{(1)}, X, \Theta, \Psi, P} f_{X \mid Y^{(1)}, \Theta, \Psi, P} f_{  Y^{(1)} \mid \Theta, \Psi, P} \\
  &=
    f_{Y^{(2)} \mid X, \Theta} f_{X \mid Y^{(1)}, \Psi, P} f_{Y^{(1)} \mid \Theta} ,
\end{aligned}
\]

so that plugging back into the expression for the posterior we obtain,

\[
\begin{aligned}
  f_{\Theta \mid X, Y, P}
  &=
    \frac{\left(\int  f_{X \mid Y^{(1)}, \Psi, P}  f_{\Psi \mid \Theta} \; d \psi \right) f_{Y^{(2)} \mid X, \Theta} f_{Y^{(1)} \mid \Theta} f_{\Theta}}{\int \left(\int f_{(X,Y) \mid \Theta, \Psi, P} f_{\Psi \mid \Theta, P} \: d \psi\right) f_{\Theta \mid P} \: d \theta }
.
\end{aligned}
\] where we have also used that \(f_{\Theta \mid P} = f_{\Theta}\) and
\(f_{\Psi \mid \Theta, P} = f_{\Psi \mid \Theta}\).

The integral in the numerator is, \[
\begin{aligned}
  \int
  &
    f_{X \mid Y^{(1)}, \Psi, P} (x \mid y^{(1)}, \psi, p)  f_{\Psi \mid \Theta} (\psi \mid \theta) \; d \psi
  \\
  &=
    \begin{cases}
      x & \text{ if } p = A \\
      \cexpect{\cproba{X = x}{Y^{(1)} = y^{(1)}, \Psi, P = B}}{\Theta = \theta} & \text{ if } p = B
    \end{cases}
\end{aligned}
\] We see for investigator \(A\) that the integral is independent of
\(\theta\), so that the posterior has the same form as in
Section~\ref{sec-q1-and-conventional-analysis}. As the sampling model
and \(\Theta\) prior is unaltered, we conclude that the posterior for
investigator A is still the normal distribution from
(\ref{eq-posterior-q1}), i.e.~the posterior density for investigator
\(A\) is, \begin{equation}\phantomsection\label{eq-posterior-q2-A}{
f_{\Theta \mid X, Y, P} (\theta \mid x, y, A)
=
\frac{1}{\sigma_{1+x}} \phi \left( \frac{\theta - \mu_{1+x}}{\sigma_{1+x}}\right)
}\end{equation} with \(x = 1\) (see (\ref{eq-q2-mu-sigma-x}) for the
expressions for \(\mu_{1+x}\) and \(\sigma_{1+x}\)).

Turning to investigator \(B\), the integral has the form
\begin{equation}\phantomsection\label{eq-q2-integral-investigator-B}{
\begin{aligned}
  \int
  &
    f_{X \mid Y^{(1)}, \Psi, P} (x \mid y^{(1)}, \psi, B)  f_{\Psi \mid \Theta} (\psi \mid \theta) \; d \psi
  \\
  &=
  \pi(\theta)^x \left[ 1 - \pi(\theta) \right]^{1-x}
\end{aligned}
}\end{equation} where, \[
\begin{aligned}
  \pi(\theta) 
  &=
  \cexpect{\cproba{\Psi \geq \frac{1}{n} \sum_{i=1}^n y_i}{Y^{(1)} = y^{(1)}, \Psi, P = B}}{\Theta = \theta}
  \\
  &=
  \left( 1 - \Phi \left( \frac{\bar{y}_{\boldsymbol{\cdot}}^{(1)} - (a + b \theta)}{\omega} \right) \right) ,
\end{aligned}
\] so that \[
\begin{aligned}
  \int
  &
    f_{X \mid Y^{(1)}, \Psi, P} (x \mid y^{(1)}, \psi, B)  f_{\Psi \mid \Theta} (\psi \mid \theta) \; d \psi
  \\
  &=
    \Phi \left( (-1)^{1-x} \frac{ (a + b \theta) - \bar{y}_{\boldsymbol{\cdot}}^{(1)} }{\omega} \right)
\end{aligned}
\] In other words, the posterior density for investigator \(B\) is
proportional to, \[
\Phi \left( (-1)^{1-x} \frac{ (a + b \theta) - \bar{y}_{\boldsymbol{\cdot}}^{(1)} }{\omega} \right)
\phi \left(\frac{\theta - \mu_{1+x}}{\sigma_{1+x}}\right) \frac{1}{\sigma_{1+x}}
\] where \((\mu_1, \sigma_1^2)\) and \((\mu_2, \sigma_2^2)\) are the
mean and variance of the normal distribution arising from updating the
prior \(f_{\Theta}\) using, respectively, the data \(Y^{(1)}\) and
\((Y^{(1)}, Y^{(2)})\) when \(b=0\) (see (\ref{eq-q2-mu-sigma-x}) for
the expressions). The normalising constant is the integral, \[
\begin{aligned}
  &
    \int_{\mathbb{R}}
    \Phi \left( (-1)^{1-x} \frac{ (a + b \theta) - \bar{y}_{\boldsymbol{\cdot}}^{(1)} }{\omega} \right)
    \phi \left(\frac{\theta - \mu_{1+x}}{\sigma_{1+x}}\right) \frac{1}{\sigma_{1+x}}
    \: d \theta \\
  \quad &=
          \expect{
          \Phi \left( (-1)^{1-x} \frac{ (a + b U) - \bar{y}_{\boldsymbol{\cdot}}^{(1)} }{\omega} \right)
          }
\end{aligned}
\] for \(U \sim N(\mu_{1+x}, \sigma^2_{1+x})\). Using the results from
\citet{demidenko_mixed_2013} Section 7.1.2, we obtain, \[
\begin{aligned}
  &
    \int_{\mathbb{R}}
    \Phi \left( (-1)^{1-x} \frac{ (a + b \theta) - \bar{y}_{\boldsymbol{\cdot}}^{(1)} }{\omega} \right)
    \phi \left(\frac{\theta - \mu_{1+x}}{\sigma_{1+x}}\right) \frac{1}{\sigma_{1+x}}
    \: d \theta \\
  \quad &=
          \Phi
          \left(
          \frac{
          (-1)^{1-x} \frac{a + b \mu_{1+x} - \bar{y}_{\boldsymbol{\cdot}}^{(1)} }{\omega}
          }{
          \sqrt{1 + \sigma_{1+x}^2 b^2/\omega^2}
          }
          \right)
          \\
  \quad &=
          \Phi
          \left(
          (-1)^{1-x}
          \frac{
          a + b \mu_{1+x} - \bar{y}_{\boldsymbol{\cdot}}^{(1)}
          }{
          \sqrt{\omega^2 + \sigma_{1+x}^2 b^2}
          }
          \right)
\end{aligned}
\] so that the posterior density is given by, \[
f_{\Theta \mid X, Y, P} (\theta \mid x, y, B)
=
\ExpressionPosteriorDensity
\] Using this, we can derive the moment generating function of the
posterior, \[
M(t) =
\cexpect{e^{t \Theta}}{X=x, Y=y, P = B}
=
\int_{\mathbb{R}} e^{t \theta} f_{\Theta \mid X, Y, P} (\theta \mid x, y, B) \: d \theta
\] Noting that, \[
\begin{aligned}
  e^{t \theta} \phi \left(\frac{\theta - \mu_{1+x}}{\sigma_{1+x}}\right) 
  &=
    \exp
    \left\{
    - \frac{1}{2 \sigma^2_{1+x}} \left(\theta - [\mu_{1+x} + t \sigma^2_{1+x}]\right)^2
    \right\}
    \exp
    \left\{
    \frac{t^2}{2} \sigma^2_{1+x} + \mu_{1+x} t
    \right\}
    \frac{1}{\sqrt{2 \pi}}
  \\
  &=
    \phi \left(\frac{\theta - [\mu_{1 + x} + t \sigma^2_{1+x}]}{\sigma_{1+x}}\right)
    \exp
    \left\{
    \frac{t^2}{2} \sigma^2_{1+x} + \mu_{1+x} t
    \right\}
\end{aligned}
\] we find, \[
\begin{aligned}
  M(t)
  &=
    \frac{
    \exp
    \left\{
    \frac{t^2}{2} \sigma^2_{1+x} + \mu_{1+x} t
    \right\}
    }{
    \Phi
    \left(
    (-1)^{1-x}
    \frac{
    a + b \mu_{1+x} - \bar{y}_{\boldsymbol{\cdot}}^{(1)}
    }{
    \sqrt{\omega^2 + \sigma_{1+x}^2 b^2}
    }
    \right)
    \sigma_{1+x}
    }
    \int_{\mathbb{R}}
    \Phi \left( (-1)^{1-x} \frac{ (a + b \theta) - \bar{y}_{\boldsymbol{\cdot}}^{(1)} }{\omega} \right)
    \phi \left(\frac{\theta - [\mu_{1 + x} + t \sigma^2_{1+x}]}{\sigma_{1+x}}\right)
    \: d \theta
    \\
  &=
    \frac{
    \exp
    \left\{
    \frac{t^2}{2} \sigma^2_{1+x} + \mu_{1+x} t
    \right\}
    }{
    \Phi
    \left(
    (-1)^{1-x}
    \frac{
    a + b \mu_{1+x} - \bar{y}_{\boldsymbol{\cdot}}^{(1)}
    }{
    \sqrt{\omega^2 + \sigma_{1+x}^2 b^2}
    }
    \right)
    }
    \mathbb{E}_{Z \sim N(0,1)}
    \left[
    \Phi \left( (-1)^{1-x} \frac{ (a + b [\mu_{1 + x} + t \sigma^2_{1+x} + \sigma_{1+x} Z ]) - \bar{y}_{\boldsymbol{\cdot}}^{(1)} }{\omega} \right)
    \right]
    \\
  &=
    \frac{
    \exp
    \left\{
    \frac{t^2}{2} \sigma^2_{1+x} + \mu_{1+x} t
    \right\}
    }{
    \Phi
    \left(
    (-1)^{1-x}
    \frac{
    a + b \mu_{1+x} - \bar{y}_{\boldsymbol{\cdot}}^{(1)}
    }{
    \sqrt{\omega^2 + \sigma_{1+x}^2 b^2}
    }
    \right)
    }
    \Phi \left( (-1)^{1-x} \frac{\frac{ (a + b [\mu_{1 + x} + t \sigma^2_{1+x} ]) - \bar{y}_{\boldsymbol{\cdot}}^{(1)} }{\omega}}{\sqrt{1 + \sigma_{1+x}^2 b^2 / \omega^2}} \right)
    \\
  &=
    \frac{
    \exp
    \left\{
    \frac{t^2}{2} \sigma^2_{1+x} + \mu_{1+x} t
    \right\}
    }{
    \Phi
    \left(
    (-1)^{1-x}
    \frac{
    a + b \mu_{1+x} - \bar{y}_{\boldsymbol{\cdot}}^{(1)}
    }{
    \sqrt{\omega^2 + \sigma_{1+x}^2 b^2}
    }
    \right)
    }
    \Phi
    \left(
    (-1)^{1-x}
    \frac{a + b [\mu_{1 + x} + t \sigma^2_{1+x} ] - \bar{y}_{\boldsymbol{\cdot}}^{(1)} }{\sqrt{\omega^2 + \sigma_{1+x}^2 b^2}}
    \right)
\end{aligned}
\] Using this we can for example derive the posterior mean as, \[
\footnotesize
\begin{aligned}
  &
    \cexpect{\Theta}{X = x, Y = y, P = B} \\
  &=
    \frac{d}{dt} M(t) \Bigr|_{t = 0}
  \\
  &=
    \left[
    \Phi
    \left(
    (-1)^{1-x}
    \frac{
    a + b \mu_{1+x} - \bar{y}_{\boldsymbol{\cdot}}^{(1)}
    }{
    \sqrt{\omega^2 + \sigma_{1+x}^2 b^2}
    }
    \right)
    \right]^{-1}
    \left\{
    \mu_{1+x}
    \Phi
    \left(
    (-1)^{1-x}
    \frac{a + b \mu_{1 + x}  - \bar{y}_{\boldsymbol{\cdot}}^{(1)} }{\sqrt{\omega^2 + \sigma_{1+x}^2 b^2}}
    \right)
    +
    \phi
    \left(
    (-1)^{1-x}
    \frac{a + b \mu_{1 + x}  - \bar{y}_{\boldsymbol{\cdot}}^{(1)} }{\sqrt{\omega^2 + \sigma_{1+x}^2 b^2}}
    \right)
    \frac{
    (-1)^{1-x}
    b \sigma^2_{1 + x}
    }{\sqrt{\omega^2 + \sigma_{1+x}^2 b^2}}
    \right\} \\
  &=
  \ExpressionPosteriorMean
\end{aligned}
\]

\section{Appendix: Stan program}\label{sec-stan-program}

\begin{lstlisting}[basicstyle=\small, language=Stan]
// Paradox
//
// Make n observations of Y1 from the sampling model N(Theta,
// \sigma^2). If mean(Y1) > Psi, stop the trial, otherwise collect
// another n observation Y2 from the same model.
//
// Prior: Theta is N(mu, tau^2) and Psi prior is defined as a location
// scale model
//	a + b*Theta + omega*epsilon
// where epsilon is standard normal.
//
// Note that taking b=0, the contribution to the posterior from X does
// not depend on Theta. This may be used to draw from posterior of
// investigator A

//////////////////////////////////////////////////////////////////////

data {
  // Observed data
  int <lower=1> n;
  int<lower=0, upper=1> X; // continuation indicator
  vector[n] Y1;
  vector[n] Y2;

  // Hyperparameters
  real sigma; // Y data standard deviation (known)
  real mu; // Mean in Theta prior 
  real tau; // Standard deviation in Theta prior
  real a; // Intercept in Psi prior
  real b; // Slope in Psi prior (dependence with Theta)
  real omega; // Standard deviation in Psi prior
}

// If X=0, only observe Y1 
transformed data {
  vector[2*n] Ymax;
  Ymax = append_row(Y1, Y2);

  int<lower=n, upper=2*n> N;
  N = n*(1 + X); // random sample size

  vector[N] Y;
  for (i in 1:N) {
    Y[i] = Ymax[i];
  }

  real meanY1 = mean(Y1);
}


parameters {
  real Theta;
}


model {
  Theta ~ normal(mu, tau);
  
  // P(X=1) = P(Psi >= meanY1) = 1 - probit((meanY1 - location)/scale)
  X ~ bernoulli(Phi(-(meanY1 - (a + b*Theta))/omega)); 

  Y ~ normal(Theta, sigma);
}

//////////////////////////////////////////////////////////////////////
// End
\end{lstlisting}

\end{document}